\documentclass[journal]{IEEEtran}

\ifCLASSINFOpdf
\else
\fi

\usepackage{amsthm}
\usepackage{amsfonts,amssymb,latexsym,cite}
\usepackage[cmex10]{amsmath}
\usepackage{algorithmic}
\usepackage{array}
\usepackage{bbm}
\usepackage{mathrsfs}
\usepackage{multirow}
\usepackage{verbatim}
\usepackage{color,xcolor}
\usepackage{graphicx}
\usepackage{epstopdf}
\usepackage{subcaption}
\usepackage[font={small,it}]{caption}
\usepackage[T1]{fontenc}
\usepackage{url}
\usepackage{enumerate}
\usepackage{hyperref} 
\graphicspath{{./figures/}} 
\usepackage{caption}
\usepackage{stfloats} 
\usepackage[blocks]{authblk}
\usepackage{epsfig}
\usepackage{latexsym}
\usepackage{soul}
\usepackage{booktabs} 
\usepackage{enumitem,kantlipsum}

\newtheorem{lemma}{Lemma}
\newtheorem{definition}{Definition}
\newtheorem{theorem}{Theorem}

\newtheorem{example}{Example}
\newtheorem{remark}{Remark}

\newcommand{\upperRomannumeral}[1]{\uppercase\expandafter{\romannumeral#1}}

\usepackage{lipsum}

\newcommand{\PP}{\mathbb{P}}
\newcommand{\E}{\mathbb{E}}

\newcommand{\ta}{\theta_{\mathrm{a}}}
\newcommand{\tr}{\theta_{\mathrm{r}}}
\newcommand{\tuv}{\theta_{u}}
\newcommand{\tuvp}{\theta_{v}}

\newcommand{\suv}{\mathcal{S}_{uv}}

\newcommand{\U}{\mathsf{U}}
\newcommand{\vij}{\mathsf{U}_{ij}}
\newcommand{\bij}{(B_{\xi'})_{ij}}
\newcommand{\bbij}{(B_{\xi})_{ij}}
\newcommand{\tar}{\tau_{\mathrm{ar}}}
\newcommand{\tarr}{\nu_{\mathrm{ar}}}

\newcommand{\taaa}{\nu_{\mathrm{aa}}}

\newcommand{\trrr}{\nu_{\mathrm{rr}}}

\newcommand{\saa}{\mathcal{S}_{\mathrm{aa}}}
\newcommand{\sar}{\mathcal{S}_{\mathrm{ar}}}
\newcommand{\sra}{\mathcal{S}_{\mathrm{ra}}}
\newcommand{\srr}{\mathcal{S}_{\mathrm{rr}}}

\newcommand{\At}{\mathcal{A}_{\mathrm{t}}}
\newcommand{\Aa}{\mathcal{A}_{\mathrm{a}}}
\newcommand{\Rt}{\mathcal{R}_{\mathrm{t}}}
\newcommand{\Ra}{\mathcal{R}_{\mathrm{a}}}

\newcommand{\ttt}{h(\theta)}
\newcommand{\zij}{Z_{ij}}
\newcommand{\V}{\mathcal{V}}
\newcommand{\e}{\mathcal{E}}
\newcommand{\tij}{\Theta_{ij}}
\newcommand{\taij}{\Theta^{\mathrm{a}}_{ij}}
\newcommand{\trij}{\Theta^{\mathrm{r}}_{ij}}
\newcommand{\lt}{f(\theta)}
\newcommand{\s}{\mathcal{S}}
\newcommand{\tX}{\widetilde{X}}
\newcommand{\tY}{\widetilde{Y}}
\newcommand{\Imm}{\mathcal{I}_{\mathrm{mm}}}
\newcommand{\Imw}{\mathcal{I}_{\mathrm{mw}}}
\newcommand{\Iwm}{\mathcal{I}_{\mathrm{wm}}}
\newcommand{\Iww}{\mathcal{I}_{\mathrm{ww}}}
\newcommand{\Iaa}{\mathcal{I}_{\mathrm{aa}}}
\newcommand{\Iar}{\mathcal{I}_{\mathrm{ar}}}
\newcommand{\Ira}{\mathcal{I}_{\mathrm{ra}}}
\newcommand{\Irr}{\mathcal{I}_{\mathrm{rr}}}

\newcommand{\mc}[1]{{\color{blue}#1}}

\allowdisplaybreaks[3]

\begin{document}
%
\title{Community Detection and Matrix Completion \\with Social and Item Similarity Graphs}
%
%
%

\author{Qiaosheng Zhang,
	Vincent Y.~F.~Tan,~\IEEEmembership{Senior~Member,~IEEE,}
	Changho Suh,~\IEEEmembership{Senior~Member,~IEEE}
	\thanks{Q.~Zhang and V.~Y.~F.~Tan are with the Department of Electrical and Computer Engineering, National University of Singapore (Emails: elezqiao@nus.edu.sg, vtan@nus.edu.sg). V.~Y.~F.~Tan is also with the Department of Mathematics, National University of Singapore. C.~Suh is with the School of Electrical Engineering, Korea Advanced Institute of Science and Technology (Email: chsuh@kaist.ac.kr).} 
	\thanks{Q.~Zhang and V.~Y.~F.~Tan are supported by Singapore National Research Foundation (NRF) Fellowship (R-263-000-D02-281) and an NUS-Berlin Alliance Grant (R-263-000-E13-133). C.~Suh is supported by the Institute of Information \& Communications Technology Planning \& Evaluation (IITP) grant funded by the Korea government (MSIT) (2020-0-00626, Ensuring high AI learning performance with only a small amount of training data)} \thanks{This paper was presented in part and virtually at the 2020 International Symposium on Information Theory (ISIT).} }

%
%

\markboth{}%
{Shell \MakeLowercase{\textit{et al.}}: Bare Demo of IEEEtran.cls for IEEE Journals}
%



\maketitle

\begin{abstract}
  We consider the problem of recovering a binary rating matrix as well as clusters of users and items based on a partially observed matrix together with side-information in the form of social and item similarity graphs. These two graphs are both generated according to the celebrated  stochastic block model (SBM). We develop lower and upper bounds on sample complexity that match for various scenarios. Our information-theoretic results quantify the benefits of the availability of the social and item similarity graphs. Further analysis reveals that under certain scenarios, the social and item similarity graphs produce an interesting synergistic effect. This means that observing two graphs is strictly better than observing just one in terms of reducing the sample complexity.
 
\end{abstract}


\section{Introduction} \label{sec:introduction}
Recommender systems aim to accurately predict users' preferences and recommend appropriate items for users based on available data that is usually scant and/or of low quality. For example, Nexflix's movie recommender system relies heavily on a partially filled {\it rating matrix} that comprises users' evaluations of movies, and various recommendation algorithms (such as collaborative filtering~\cite{goldberg1992using,sarwar2001item,linden2003amazon,mnih2008probabilistic}) have been developed. However, merely exploiting the available ratings may not be sufficient for producing high-quality recommendations, especially when we would like to (i) recommend items to new users who have not rated any items (i.e., the {\it cold start problem}); and (ii) promote new items that have not received any ratings yet (i.e., the {\it dual cold start problem}). 

An effective approach to overcome the aforementioned challenges is to make proper use of the \emph{side-information}~\cite{jamali2010matrix, ma2011recommender,kalofolias2014matrix, ahn2018binary, yoon2018joint, jo2020discrete} such as the {\it similarities of users} (e.g., the friendships in Facebook) or {\it similarities of items} (e.g., the categories/genres of movies in the Netflix database). The rationale is that users in the same community tends to share similar preferences (called {\it homophily}~\cite{mcpherson2001birds} in the social sciences), and items of similar features are more likely to have similar attractiveness to users. 
 

Most of the prior works studied the algorithmic developments of the graph-aided recommender systems (see~\cite{tang2013social} for a review of {\it social recommender systems}); however, a relatively fewer number of works focused on the fundamental limits of such problems, which, we believe, are equally pertinent. Ahn et al.~\cite{ahn2018binary} considered the problem of recovering the binary rating matrix based on a partially observed matrix and a social graph. The authors characterized a sharp threshold on the sample complexity for recovery, and also quantified the gains due to the graph side-information. 
In addition to the social graph, an item similarity graph is also often readily available, especially in this era of massive data. For instance,
\begin{enumerate}
	\item One can make use of high-dimensional  feature vectors of items (e.g., genre or language for movies or level of calories and carbonation for beverages~\cite{Condliff1999}) to \emph{construct} item graphs via clustering algorithms~\cite{rattigan07};
	\item One can also leverage {\em users' behavior history} to construct item graphs as has been  done by Taobao~\cite{wang18billion}.
\end{enumerate}
A natural question is whether the   item similarity graph yields a {\em strict} benefit for recovery, and whether observing {\em two} pieces of graph side-information has a  {\em synergistic} effect in reducing the sample complexity. This work addresses these questions  and uncovers the roles and benefits of the social and item similarity graphs. 

We consider a movie recommender system with $n$ users and $m$ movies, and users' ratings to movies are either $0$ (dislike) or $1$ (like). For simplicity, users are partitioned into men and women, while movies are partitioned into action movies and romance movies. From anecdotal evidence, action movies usually attract more men than women, while the reverse is true for romance movies. To capture this phenomenon, we put forth the following two models.

\begin{enumerate}[label=(\roman*), wide, labelwidth=!, labelindent=0pt]
	\item We assume that men's \emph{nominal} ratings to action and romance movies are respectively `$1$' and `$0$', and women's nominal ratings to action and romance movies are respectively `$0$' and `$1$'. This assumption is shown in Table~\ref{table:1}. The corresponding model is referred to as \emph{the basic model} (Model 1).
	
	\item In reality, there may often exist \emph{atypical} movies such that their attractiveness is different from the typical ones (such as the popular action movie \emph{Captain America}, which has  a large following of  female fans, arguably more so than male fans). To capture this, we introduce another model that we call Model 2 wherein we allow the existence of {\it atypical} action and romance movies, and assume that atypical action (resp.\ romance) movies attract more women (resp.\ men). The nominal ratings for Model 2 are  shown in Table~\ref{table:2}. 
\end{enumerate}
Given clusters of users and movies, one can form an $n \times m$ \emph{nominal rating matrix} that comprises $n$ users' ratings to all the $m$ movies, according to either Model 1 or  Model 2. Each user may have distinct taste such that his/her \emph{actual ratings} may differ from the nominal ratings of the associated cluster. To model this flexibility, we assume the actual rating from each user to each movie is a perturbed version of the nominal rating. The perturbation can be viewed as \emph{personalization}, and the perturbed matrix is called the \emph{personalized rating matrix}.

\begin{table}
	\caption{Nominal ratings from users to movies (Model 1)}
	\label{table:1}
	\centering
	\begin{tabular}{lll}
		\toprule
		& Action movies     & Romance movies \\
		\midrule
		Men & \multicolumn{1}{c}{$1$} & \multicolumn{1}{c}{$0$}    \\
		Women     & \multicolumn{1}{c}{$0$} & \multicolumn{1}{c}{$1$}      \\
		\bottomrule
	\end{tabular}
\end{table}

\begin{table}
	\caption{Nominal ratings from users to movies (Model 2)}
	\label{table:2}
	\centering
	\begin{tabular}{lllll}
		\toprule
		&\multicolumn{2}{c}{Action movies}  & \multicolumn{2}{c}{Romance movies}                  \\
		\cmidrule(r){2-3} \cmidrule(r){4-5}
	     & Typical     & Atypical & Typical & Atypical \\
		\midrule
		Men & \multicolumn{1}{c}{$1$}  & \multicolumn{1}{c}{$0$}  & \multicolumn{1}{c}{$0$} & \multicolumn{1}{c}{$1$}    \\
		Women     & \multicolumn{1}{c}{$0$} & \multicolumn{1}{c}{$1$}  & \multicolumn{1}{c}{$1$} & \multicolumn{1}{c}{$0$}    \\
		\bottomrule
	\end{tabular}
\vspace{-10pt}
\end{table}

In summary, three pieces of information are observed: (i) entries in the personalized rating matrix that are sampled independently with a certain probability, (ii) the social graph, and (iii) the movie graph (see Fig.~\ref{fig:model} for a pictorial representation of our setting). The task here is to exactly recover the clusters of users and movies and to reconstruct the nominal rating matrix.

\subsection{Main Contributions} \label{sec:contribution}
In this work, we model the social and movie graphs by a celebrated generative model for random graphs---the {\it stochastic block model} (SBM)~\cite{holland1983stochastic}. For Model 1, we develop a sharp threshold on the  sample complexity for recovery. For Model 2, lower and upper bounds on the  sample complexity  are derived---they match for a wide range of parameters of interest, and match up to a factor of two for the remaining regime. Both the threshold (characterized under Model 1) and the upper and lower bounds (intended for Model 2) are functions of the \emph{qualities} of the social and movie graph. Roughly speaking, the qualities can be quantified by the difference between the intra- and inter-cluster probabilities of the SBMs that govern them. Our theoretical studies show that the sample complexity gains due to the social and movie graphs appear for a wide range of parameters.  More interestingly, we show that there exists a certain regime in which there is  a synergistic  effect generated by the two graphs---observing {\it both} graphs is strictly better than observing {\it only one} graph. 

This synergistic  effect can be seen from Fig.~\ref{fig:1}, which considers Model 1 with equal numbers of users and movies (i.e., $n = m$). It plots the sample complexity as a function of $I_1$ (the quality of social graph to be defined in Section~\ref{sec:result}) under three different values of $I_2$ (the quality of movie graph to be defined similarly). $I_1 > 0$ and $I_1 = 0$ respectively mean that the social graph is available and unavailable. Compared to the case when no graph is available (i.e., $I_1 = I_2 = 0$), the sample complexity is reduced only when both $I_1$ and $I_2$ are positive, while the gain disappears when either $I_1$ or $I_2$ becomes zero. On the other hand, if the number of users exceeds the number of movies (e.g., $n = 2m$, as illustrated in Fig.~\ref{fig:2}), the availability of social graph is always helpful in reducing the sample complexity regardless of the availability of movie graph; while the movie graph is helpful only when the quality of social graph is good enough (i.e., $I_1 > I^*$). Thus, observing two graphs with $I_1 > I^*$ and $I_2 > 0$ also produces a synergistic  effect. The reasons are provided in Section~\ref{sec:result_simple}.    

\begin{figure}
	\begin{subfigure}[b]{0.24\textwidth}
		\includegraphics[width=\textwidth]{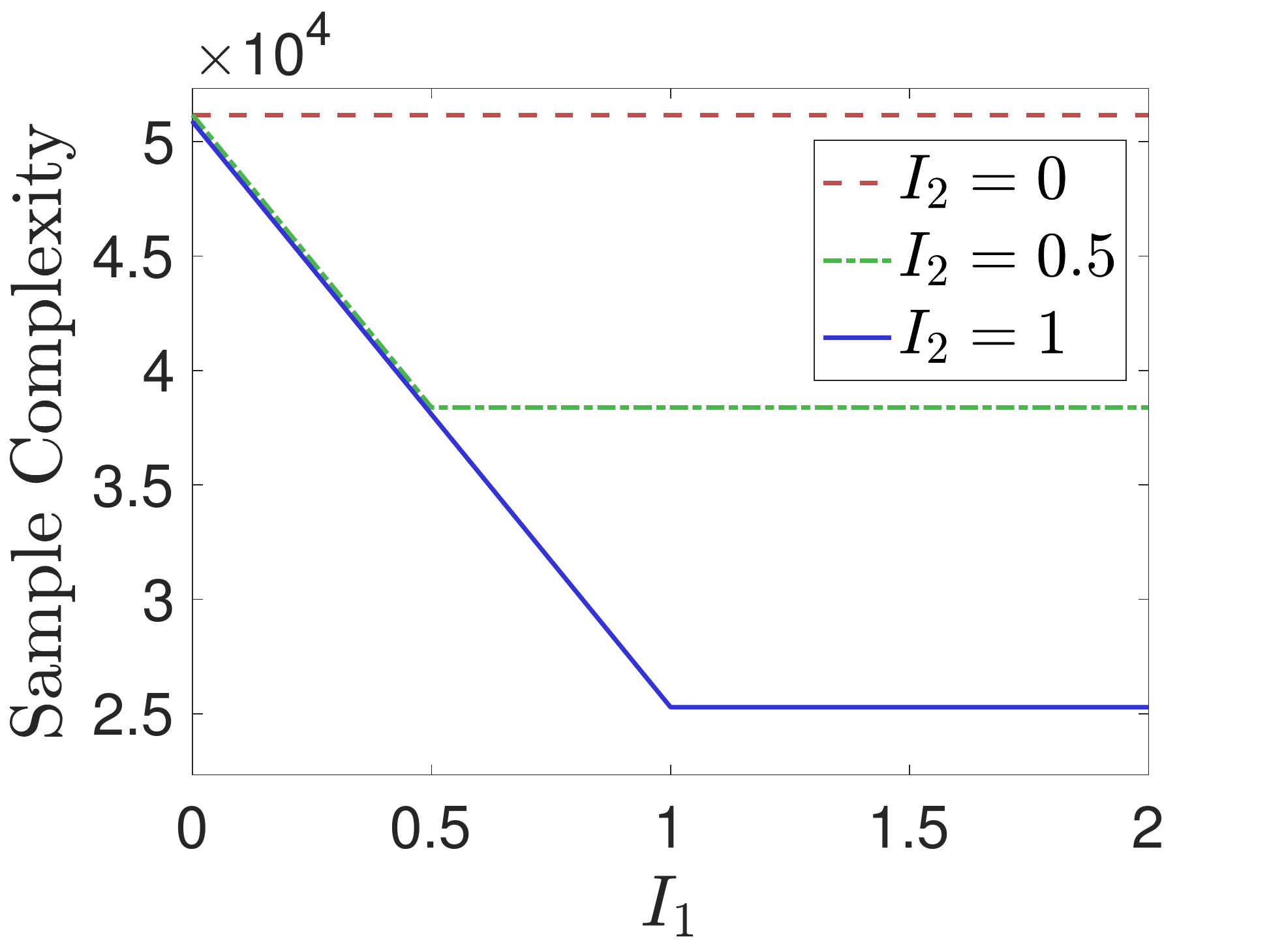}
		\caption{$n = m = 10,000$}
		\label{fig:1}
	\end{subfigure}
	\begin{subfigure}[b]{0.24\textwidth}
		\includegraphics[width=\textwidth]{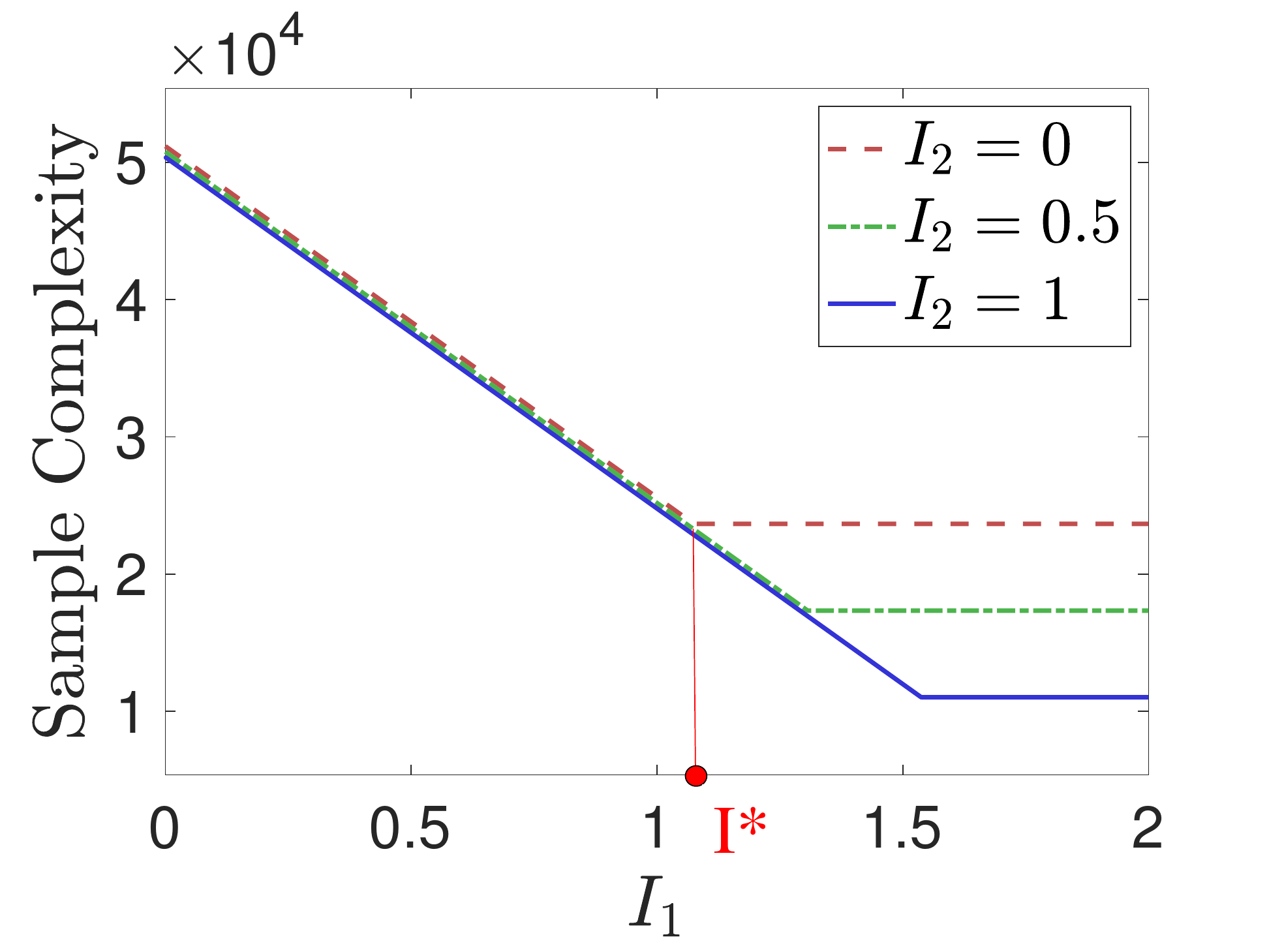}
		\caption{$n = 2m = 10,000$}
		\label{fig:2}
	\end{subfigure}
\caption{Sample complexity versus $I_1$ (Model 1).}
\label{fig:12}
\vspace{-10pt}
\end{figure}

\subsection{Related Works}
This work is closely related to community detection and matrix completion. While there is a vast literature on these two topics (especially from algorithmic and experimental perspectives), in the following, we mainly discuss theoretical works that provide provable guarantees.  

The theoretical underpinnings of community detection have been well-studied and sharp thresholds for exact recovery of communities have been successively established~\cite{abbe2015exact,mossel2015consistency, abbe2015community, hajek2017information,abbe2017community}. Moreover, it has been shown that side-information (e.g., node values~\cite{saad2018community, saad2018recovering, saad2018exact, yang2013community, mayya2019mutual}, edge weights~\cite{aicher2013adapting}, similarity information between data points~\cite{mazumdar2017query}) is also helpful in recovering communities. In our   setting, given  realizations from two SBMs together with a partially observed matrix, we are required to recover  the communities of users and movies and the rating matrix. We note that the task in~\cite{xu2014jointly} (joint recovery of rows and columns communities) is similar to ours, but therein,   graph information is not available. Another relevant problem is the \emph{labelled or weighted SBM problem}~\cite{heimlicher2012community,xu2014edge,lelarge2015reconstruction,xu2020opt, yun2016optimal}; we provide a detailed discussion on this point in Remark~\ref{remark:1} in Section~\ref{sec:result}.   


Various efficient algorithms have been developed for low-rank matrix completion~\cite{candes2009exact,candes2010power,marjanovic2012l_q,chen2015signal,dai2011subspace,ma2014decomposition}. In addition to the low-rank property, some other works considered applications in which the  matrix to be recovered has certain extra properties. In particular,~\cite{ahn2018binary,yoon2018joint, jo2020discrete} assumed that the social graph imposes dependencies amongst the rows of low-rank matrix. The task in this work can also be regarded as a matrix completion problem in which the social and movie graphs impose dependencies on both rows and columns of matrix.

Finally, we point out that the objective of this work is to gain a theoretical understanding of the benefits of graph side-information by investigating the fundamental limits of the recovery problem of interest. In contrast, a follow-up work~\cite{zhang2020mc2g} considers the same problem, but the main focus there is the design and analysis of efficient algorithms.  

\subsection{Outline}
We describe our model in Section~\ref{sec:model}. In Section~\ref{sec:result}, we present the main results for both Model 1 (Theorem~\ref{thm:simple}) and Model 2 (Theorems~\ref{thm:1} and~\ref{thm:2}), and reveal the benefits of   the social and movie graphs. Section~\ref{sec:simple}  provides the detailed proofs of Theorem~\ref{thm:simple}, while Section~\ref{sec:model2}  provides the proof sketches of Theorems~\ref{thm:1} and~\ref{thm:2}. Section~\ref{sec:conclusion} concludes this work
and proposes   directions for future research. 

\section{Problem statement} \label{sec:model}

For any integer $a \ge 1$, $[a]$ represents the set $\{1,\ldots, a\}$. For any integers $a, b$ such that $a < b$, $[a: b]$ represents  $\{a,a+1,\ldots, b\}$. For any event $\mathcal{E}$, the conditional probability $\PP(\cdot|\mathcal{E})$ is abbreviated as $\PP_{\mathcal{E}}(\cdot)$.  Throughout this paper we use standard {\it asymptotic notations}~\cite[Ch. 3.1]{leiserson2001introduction} to describe the limiting behaviour of functions/sequences. 

\subsection{Models}
Consider $n$ users and $m$ movies. To convey the main message (i.e., uncover the benefits of the social and movie graphs) as concisely as possible, we assume both users and movies are partitioned into equal-sized clusters.  The sets of men and women are respectively denoted by $\mathcal{M}$ and $\mathcal{W}$, where $\mathcal{M}, \mathcal{W} \subset [n]$, $|\mathcal{M}| = |\mathcal{W}| = n/2$, and $\mathcal{M} \cap \mathcal{W} = \emptyset$. The sets of action and romance movies are respectively denoted by $\mathcal{A}$ and $\mathcal{R}$, where $\mathcal{A}, \mathcal{R} \subset [m]$, $|\mathcal{A}| = |\mathcal{R}| = m/2$, and $\mathcal{A} \cap \mathcal{R} = \emptyset$. Without loss of generality,\footnote{Without this assumption, for any instance $(\mathcal{M},\mathcal{W}, \mathcal{A}, \mathcal{R})$, one can always find another instance $(\mathcal{M}',\mathcal{W}', \mathcal{A}', \mathcal{R}')$ with $\mathcal{W}' = \mathcal{M}$ and $\mathcal{R}' = \mathcal{A}$ (i.e., simultaneously flipping the clusters of users and movies) such that these two instances are statistically indistinguishable.} we assume that the majority of the first $n/2$ users are men (i.e., $|\mathcal{M} \cap [n/2]| \ge n/4$), and the majority of the first $m/2$ movies are action movies (i.e., $|\mathcal{A} \cap [m/2]| \ge m/4$).   
\subsubsection{Model 1 (The Basic Model)}
We assume that men's nominal rating to action and romance movies are respectively `$1$' and `$0$', and women's nominal rating to action and romance movies are respectively `$0$' and `$1$'. See Table~\ref{table:1}. For each   user, the actual rating to each movie is perturbation of the nominal rating as per  $\text{Bern}(\theta)$, where $\theta \in (0,\frac{1}{2})$ is referred to as the \emph{personalization parameter} and is independent of $m$ and $n$.

Let $\xi_{\mathcal{M},\mathcal{W},\mathcal{A},\mathcal{R}}$ be an aggregation of the parameters; these include the sets of users $\mathcal{M}$ and $\mathcal{W}$, and the sets of movies $\mathcal{A}$ and $\mathcal{R}$. We sometimes abbreviate $\xi_{\mathcal{M},\mathcal{W},\mathcal{A},\mathcal{R}}$ as $\xi$ for notational convenience. The sets of men, women, action and romance movies (associated with $\xi$) are respectively denoted by $\xi_{\mathcal{M}}, \xi_{\mathcal{W}}, \xi_{\mathcal{A}},$ and $\xi_{\mathcal{R}}$.  The parameter space $\Xi$ is the collection of valid parameters $\xi_{\mathcal{M},\mathcal{W},\mathcal{A},\mathcal{R}}$.
Given any $\xi \in \Xi$, one can construct the $n \times m$ nominal rating matrix $B_{\xi}$ based on $\xi_{\mathcal{M}}, \xi_{\mathcal{W}}, \xi_{\mathcal{A}}$, $\xi_{\mathcal{R}}$ and Table~\ref{table:1}.
The $n \times m$ personalized rating matrix $V_{\xi}$ denotes users' actual ratings to all the movies. Specifically, the $i$-th row of $V_{\xi}$ represents the $i$-th user's ratings to all the movies and  the $j$-th column of $V_{\xi}$ represents all the users' rating to the $j$-th movie. Each element $(V_{\xi})_{ij} = (B_{\xi})_{ij} \oplus \tij$, where $\{\tij\} \stackrel{\text{i.i.d.}}{\sim} \text{Bern}(\theta)$.

\subsubsection{Model 2 (The Model with Atypical Movies)}
This model assumes that there exist an \emph{unknown-sized} subset of atypical action
movies $\Aa \subseteq \mathcal{A}$ and an unknown-sized subset of atypical romance movies $\Ra \subseteq \mathcal{R}$. We refer to $\At \triangleq \mathcal{A} \setminus \Aa$ and $\Rt \triangleq \mathcal{R} \setminus \Ra$ as \emph{typical} action and romance movies, respectively. The nominal ratings from users to movies are shown in Table~\ref{table:2}, which reflects our assumption that typical action movies and atypical romance movies attract more men than women, while typical romance movies and atypical action movies attract more women than men.

For each user, it is assumed that the actual rating to each action movie is a perturbation of the nominal rating by   $\text{Bern}(\ta)$, and the actual rating to each romance movie is perturbation of the nominal rating by $\text{Bern}(\tr)$, where $\ta,\tr \in (0,\frac{1}{2})$   are independent of $m$ and $n$. We find the difference between $\ta$ and $\tr$ is an important statistic for distinguishing action and romance movies in Model 2---this will be apparent  in Theorems~\ref{thm:1} and~\ref{thm:2}.

Let $\xi_{\mathcal{M},\mathcal{W},\At,\Aa,\Rt,\Ra}$ (abbreviated as $\xi$) be an aggregation of the parameters of interest, and the sets of typical/atypical action and romance movies (associated with $\xi$) are denoted by $\xi_{\At}, \xi_{\Aa}, \xi_{\Rt}, \xi_{\Ra}$.  The parameter space $\Xi$ is the collection of valid parameters $\xi_{\mathcal{M},\mathcal{W},\At,\Aa,\Rt,\Ra}$. Given any $\xi \in \Xi$, one can construct the $n \times m$ nominal rating matrix $B_{\xi}$ based on $\xi_{\mathcal{M}}, \xi_{\mathcal{W}}, \xi_{\At}, \xi_{\Aa}, \xi_{\Rt}, \xi_{\Ra}$ and Table~\ref{table:2}. Let $\{\taij\} \stackrel{\text{i.i.d.}}{\sim} \text{Bern}(\ta)$ and $\{\trij\} \stackrel{\text{i.i.d.}}{\sim} \text{Bern}(\tr)$.
Each element of the $n \times m$ personalized rating matrix $V_{\xi}$ takes the form $(V_{\xi})_{ij} = (B_{\xi})_{ij} \oplus \taij$ if movie $j$ is an action movie, and $(V_{\xi})_{ij} = (B_{\xi})_{ij} \oplus \trij$ if movie $j$ is a romance movie.

\begin{figure}[t]
	\centering
	\includegraphics[width=6.5cm]{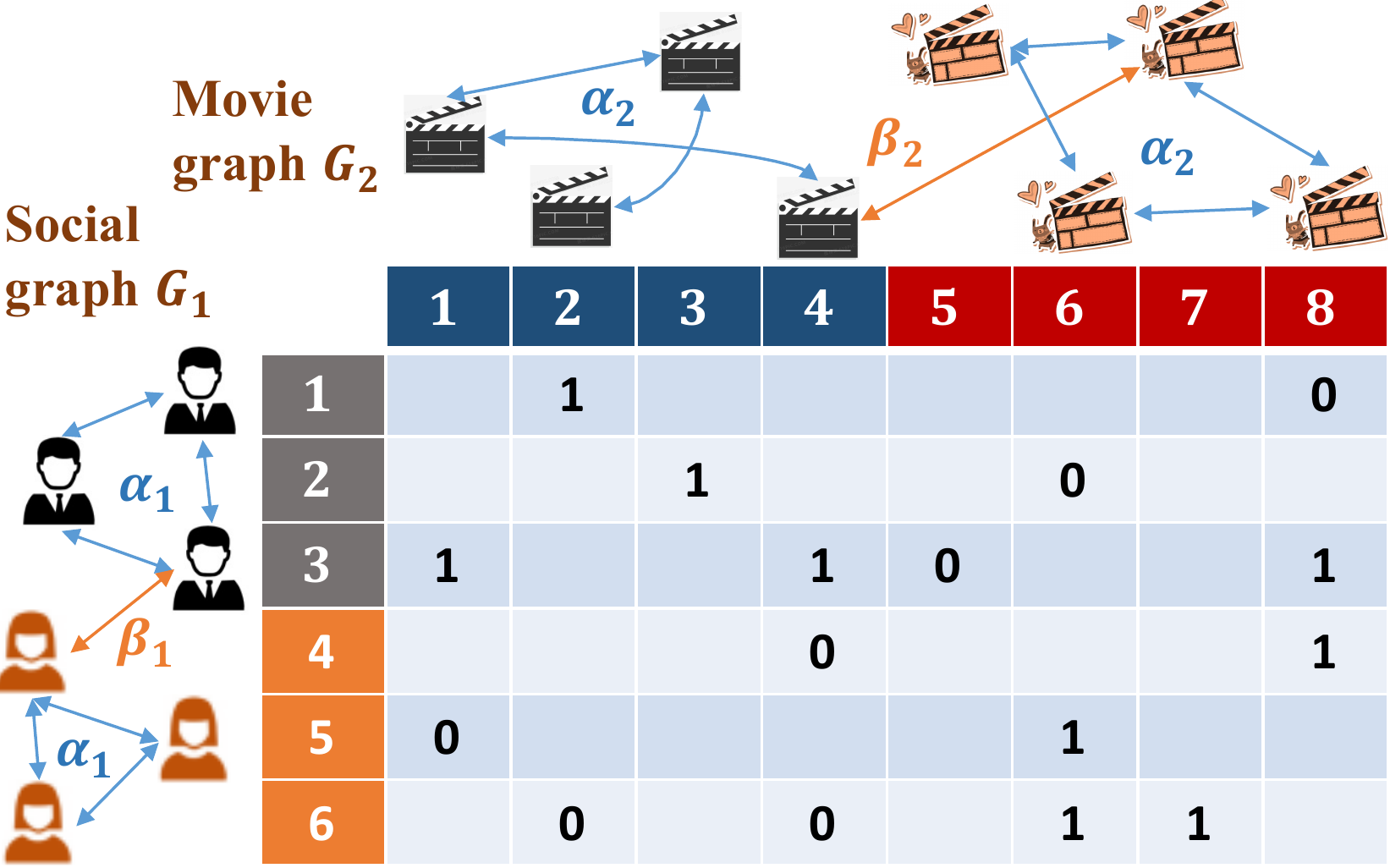}
	\caption{An illustration of $(\U, G_1,G_2)$ that is generated according to the model parameterized by $\xi$, where $\xi_\mathcal{M} = \{1,2,3\}$ (gray), $\xi_\mathcal{W} = \{4,5,6\}$ (orange), $\xi_\mathcal{A} = \{1,2,3,4\}$ (blue), $\xi_\mathcal{R} = \{5,6,7,8\}$ (red). }  
	\label{fig:model}
	\centering
	\vspace{-10pt}
\end{figure}

\subsection{Observations}
In both models, for any $\xi \in \Xi$, the learner observes the following three pieces of information.

\begin{enumerate}[wide, labelwidth=!, labelindent=0pt]
\item The partially observed matrix $\U$. Let $\{\zij \} \stackrel{\text{i.i.d.}}{\sim} \text{Bern}(p)$. For each $(i,j) \in [n] \times [m]$, the $(i,j)$-th entry of $\U$ is  
$$\U_{ij} =\begin{cases} (V_{\xi})_{ij}, \ &\text{ if } \zij = 1,\\
\perp, \ &\text{ if } \zij = 0,
\end{cases}$$
where  $\perp$ is the erasure symbol, and $p$ is the {\it sample probability}.
   
\item The social graph $G_{1} = (\V_1, \e_1)$ with $\V_1 = [n]$ being the set of $n$ user nodes. For any pairs of nodes $i \ne i'$, it is connected with probability $\alpha_1$ if both $i$ and $i'$ are in the same cluster (either in  $\xi_{\mathcal{M}}$ or $\xi_{\mathcal{W}}$), and is connected with probability $\beta_1$ otherwise. 

\item The movie graph $G_{2} = (\V_2, \e_2)$ with $\V_2 = [m]$ being the set of $m$ movie nodes. For any pairs of nodes $j \ne j'$, it is connected with probability $\alpha_2$ if both $j$ and $j'$ are in the same cluster (either in $\xi_{\mathcal{A}}$ or $\xi_{\mathcal{R}}$), and is connected with probability $\beta_2$ otherwise. 
\end{enumerate}

An example of the three pieces of information $(\U, G_1,G_2)$ is illustrated in Fig.~\ref{fig:model}. Throughout this work, we assume $m = \omega(\log n)$ and $n = \omega(\log m)$ such that $m\to\infty$ as $n\to\infty$. 

\begin{remark}{\em
In Model 2,	one may alternatively treat typical action movies $\At$, atypical action movies $\Aa$, typical romance movies $\Rt$, and atypical romance movies $\Ra$ as four {\em distinct} clusters. However, the movie graph cannot be represented as a \emph{general SBM}~\cite{abbe2015community} with these four clusters. This is because the relative sizes   $\At, \Aa, \Rt, \Ra$ are \emph{unknown} (as the number of atypical movies is unknown), while in a general SBM the relative sizes of the clusters are predefined. 

Besides, if $\At$ and $\Aa$ (resp.\ $\Rt$ and $\Ra$) are viewed as two sub-clusters of the cluster of action movies $\mathcal{A}$ (resp.\ romance movies $\mathcal{R}$), then  there would be a  \emph{hierarchy}. However, it is worth pointing out some subtleties: (i) As mentioned above, the relative sizes of the sub-clusters are not known;  (ii) The intra-cluster probability is exactly the same as the inter-cluster probability between $\At$ and $\Aa$ (or between $\Rt$ and $\Ra$), which seems to be an unrealistic assumption for random graphs with hierarchical structures. For hierarchical graph side information in the matrix completion problem, the reader may wish to refer to the recent work~\cite{elmahdy2020matrix}.}
\end{remark}

\subsection{Objective} \label{sec:objective}

Given $(\U, G_1,G_2)$, the learner constructs an estimator $\phi=\phi(\U, G_1,G_2)$ to recover $\xi$, which includes the clusters of users ($\xi_\mathcal{M}$ and $\xi_\mathcal{W}$) and  the clusters of movies ($\xi_\mathcal{A}$ and $\xi_\mathcal{R}$ in Model 1; $\xi_{\At}, \xi_{\Aa}, \xi_{\Rt}, \xi_{\Ra}$ in Model 2). If an estimator is able to recover $\xi$ reliably, it is also able to reliably recover the nominal rating matrix $B_{\xi}$. That is, matrix completion comes as a consequence of recovering $\xi$, since one can construct the $n \times m$ matrix $B_{\xi}$ based on clusters of users and movies.

\begin{definition}[Exact recovery]
	For any estimator $\phi$, the maximum error probability is defined as 
	\begin{align}
	P_{\emph{err}}(\phi) \triangleq \max_{\xi \in \Xi}\PP_\xi(\phi(\U, G_1,G_2) \ne \xi), \label{eq:error}
	\end{align}
	where $\PP_\xi(\cdot)$ is the error probability when $(\U, G_1,G_2)$ is generated as per the model parameterized by $\xi$.
	A sequence of estimators $\Phi = \{\phi_n \}_{n=1}^{\infty}$ ensures exact recovery if 
	\begin{align}
	\lim_{n \to \infty} P_{\emph{err}}(\phi_n) = 0. \label{eqn:vanish_prob}
	\end{align}  
\end{definition}

\begin{remark}{\em
 For Model 1, recovering the nominal rating matrix is also sufficient for the recovery of $\xi$. However, this is not true for Model 2, since movies that attract more men may be regarded as either typical action or atypical romance movies.}
\end{remark}

\begin{definition}[Sample complexity]
The {\em sample complexity}  is the infimum of the expected number of entries in the matrix $\U$ such that there exists $\Phi$ for which~\eqref{eqn:vanish_prob} holds.
\end{definition}
We remark that the sample complexity can also be expressed as $nmp^*$ where $p^*$  is the {\em minimum sample probability} (MSP)  such that  $ P_{\text{err}}(\phi_n) \to 0$ as $n $ grows.

\section{Main results} \label{sec:result}
As mentioned in Section~\ref{sec:introduction}, our main contribution is to characterize the   sample complexity. These quantities are functions of the qualities of graphs, which are defined as follows.
\begin{itemize}
	\item Let $I_{1} \triangleq n(\sqrt{\alpha_1}-\sqrt{\beta_1})^2/\log n$ be a measure of the quality of the social graph $G_1$. Intuitively, a larger $I_1$ implies a larger difference between $\alpha_1$ and $\beta_1$; this   means that the structure of communities are more transparent. Thus, increasing $I_1$ makes it easier to recover the communities of users.  A well-known result~\cite{abbe2015exact, mossel2015consistency} states that it is possible to recover $\mathcal{M}$ and $\mathcal{W}$ (based on the observation of $G_1$ only)  if $I_1 > 2$, and impossible if $I_1 <2$.
	
	\item Analogously, let $I_{2} \triangleq m(\sqrt{\alpha_2}-\sqrt{\beta_2})^2/\log m$ be the quality of the movie graph $G_2$.
\end{itemize}
For ease of presentation, we state our results in terms of the sample probability $p$, instead of the sample complexity. 

\subsection{Model 1 (The Basic Model)} \label{sec:result_simple}
Theorem~\ref{thm:simple} provides a sharp threshold on the sample probability $p$, as a function of $n,m$, $I_1$, $I_2$, and the personalization parameter $\theta$. Let $h(x) \triangleq (\sqrt{1-x} -\sqrt{x})^2$.

\begin{theorem} \label{thm:simple}
Consider any $\epsilon > 0$. If
	\begin{align}
	p \!\ge \!\max\!\left\{\!\frac{(2(1\!+\!\epsilon) - I_{1}) \log n}{2\ttt m}, \frac{(2(1\!+\!\epsilon) - I_{2})\log m}{2\ttt n} \!\right\}, \label{eq:achievability}
	\end{align}
 then there exists a sequence of estimators $\Phi = \{\phi_n\}_{n=1}^\infty$ satisfying $\lim_{n \to \infty} P_{\emph{err}}(\phi_n) = 0$. If
	\begin{align}
	p \!\le\! \max\!\left\{\!\frac{(2(1\!-\!\epsilon) - I_{1}) \log n}{2\ttt m}, \frac{(2(1\!-\!\epsilon) - I_{2})\log m}{2\ttt n} \!\right\}, \label{eq:converse}
	\end{align}
 then $\lim_{n \to \infty} P_{\emph{err}}(\phi_n) = 1$ for any sequence of estimators $\Phi$.
\end{theorem}
The proof of Theorem~\ref{thm:simple} is presented in Section~\ref{sec:simple}. For the achievability part in~\eqref{eq:achievability}, the estimator $\Phi$ is chosen to be the {\it maximum likelihood (ML) estimator} $\phi_{\text{ML}}$. The converse presented in~\eqref{eq:converse} is the so-called  \emph{strong converse}~\cite{wolfowitz2012coding}. It states that as long as $p$ is smaller than the threshold in~\eqref{eqn:p_star}, the error probability of {\em any} estimator asymptotically goes to \emph{one}.  

Some additional remarks are in order.

\begin{enumerate}[wide, labelwidth=!, labelindent=0pt]
\item  Theorem~\ref{thm:simple} implies that the MSP is 
\begin{equation}
	p^*= \max \left\{ \frac{(2  - I_{1}) \log n}{2\ttt m}, \frac{(2  - I_{2})\log m}{2\ttt n} \right\}. \label{eqn:p_star}
	\end{equation}
When $2-I_1$ and $2-I_2$ are positive and do not scale with $m$ and $n$ respectively, the sample complexity is of the order $\Theta(\max\{n \log n, m \log m\})$. 
Intuitively speaking, the first term in the right-hand side (RHS) of~\eqref{eqn:p_star} is the threshold for recovering clusters of users, while the second term in the RHS of~\eqref{eqn:p_star} is the threshold for recovering clusters of movies.
  In fact, from~\eqref{eq:achievability} and~\eqref{eq:converse}, we see that if $n\log n=c m\log m$ for some $c>1$,\footnote{This is a practically relevant regime since the number of users usually exceeds the number of movies.} the {\em normalized sample complexity} is 
	\begin{equation}
	 \frac{ nmp^* }{ n\log n  }= \max\left\{ \frac{2-I_1}{2\, h(\theta)} , \frac{2-I_2}{2\, c \, h(\theta)}   \right\}. 
	  \end{equation}  

\item Recall from standard results in community detection~\cite{abbe2015exact, mossel2015consistency} that one can recover the clusters of users (resp.\ movies) based on the social  (resp.\ movie) graph   only when $I_1 > 2$ (resp.\ $I_2 > 2$). Hence, when both $I_1 > 2$ and $I_2 > 2$, samples of the rating matrix are no longer needed. This is also reflected in our result---the MSP $p^* = 0$ when both $I_1 > 2$ and $I_2 > 2$.	  
	  
\item The MSP in~\eqref{eqn:p_star} is an increasing function of the personalization parameter $\theta$. This dovetails with our intuition since more samples are needed if there are more ratings  deviating from the nominal ones.
		  
\item The availability of  graphs (manifested in positive $I_1$ and $I_2$) in general helps to reduce the sample complexity.
\begin{enumerate}[label=(\roman*)]
\item  When $n = m$, we highlighted in Section~\ref{sec:contribution} that observing both graphs helps to reduce the MSP, while observing only one graph is equivalent to observing neither; thus the availability of {\em both} graphs produces a synergistic effect. This is because in the absence of graphs (i.e., $I_1 = I_2 = 0$), the first and the second terms in~\eqref{eqn:p_star} are equal, implying that recovering the clusters of users and movies require the same number of samples. Thus, observing both graphs reduces the MSP, which makes intuitive sense because the social graph (resp. movie graph) is helpful in clustering users (resp. movie). In contrast, only having the social graph fails to reduce
the number of samples needed for clustering movies. 

\item When $n > m$ (as illustrated in Fig.~\ref{fig:2}), the availability of the social graph $G_1$ always helps to reduce the MSP. This is because in the absence of graphs, the MSP in~\eqref{eqn:p_star} is dominated by the first term (i.e., more samples are needed to recover the clusters of users than to recover the clusters of movies). Thus, having a positive $I_1$ reduces the MSP, and the MSP decreases \emph{linearly} with $I_1$ when $I_1 \le I^*$ (i.e., $I_1$ is sufficiently small such that the first term in~\eqref{eqn:p_star} is dominant). We also note  that when $I_1 > I^*$ and $I_2 = 0$ (i.e., $G_1$ is of sufficiently high quality but the movie graph $G_2$ is unavailable), the sample complexity gain is ``saturated''. This is because recovering the clusters of users is no longer the dominant task; instead, more samples are required to recover the clusters of movies (or equivalently, the second term in~\eqref{eqn:p_star} becomes  dominant). As a consequence, the movie graph helps to further reduce the MSP. Therefore, observing two graphs (with $I_1 > I^*$ and $I_2 > 0$) is strictly better than observing only one graph.
 \item When $m > n$, observing the movie graph $G_2$ helps to reduce the MSP, and the social graph is helpful only when $G_2$ is of sufficiently high quality. The reason is similar and symmetric to case (ii).
\end{enumerate}
\end{enumerate}

\begin{remark}{\em
As mentioned in the introduction, Ahn {\em et al.}~\cite{ahn2018binary} considered a binary matrix completion problem with a single social graph. A key message therein is that the social graph helps to reduce the sample complexity when the number of users is relatively large compared to that of the items, and does not help otherwise. We note that Model 1 is similar (but not identical) to the model in~\cite{ahn2018binary} when the movie graph is unavailable (i.e., $I_2 = 0$). As discussed earlier and illustrated in Fig.~\ref{fig:12}, the social graph is helpful only when the number of users is larger than that of movies; this coincides with the key message in~\cite{ahn2018binary} at a high level. }
\end{remark}

\begin{remark} \label{remark:1} {\em 
Model 1 is   related to the so-called labelled or weighted SBM problem~\cite{heimlicher2012community,xu2014edge,lelarge2015reconstruction,xu2020opt, yun2016optimal}, in which different labels are assigned to different edges \emph{probabilistically}. To see this, one can map the two symmetric SBMs into a single unified SBM that consists of all the user nodes and movie nodes, and rating information can be viewed as labels between user and movie nodes.   One major distinction  is that prior works assume that the sizes of communities scale {\em linearly} with one another, whereas this work assumes the communities are of sizes $\Theta(n)$ and $\Theta(m)$, and $n$ and $m$ are allowed to tend to infinity at different rates, subject to  $m = \omega(\log n)$ and $n = \omega(\log m)$. Model 2 below, however, is completely different from the labelled or weighted SBM problem and, as we mentioned, motivated by our desire to situate our models closer to real world settings.}
\end{remark}

\subsection{Model 2 (The Model with Atypical Movies)}
Recall that the personalization parameters for action and romance movies are respectively $\ta$ and $\tr$.
We first define two functions of $\ta$ and $\tr$ as follows:
\begin{align*}
\tau_{uv} &\triangleq 1 \!-\! \sqrt{\tuv\tuvp} \!-\! \sqrt{(1-\tuv)(1-\tuvp)}, \text{ for } u,v \!\in \!\{\mathrm{a},\mathrm{r}\}, \\*
\nu_{uv} &\triangleq 1 \!-\! \sqrt{\tuv(1-\tuvp)} \!-\! \sqrt{\tuvp(1-\tuv)}, \text{ for } u,v \!\in\! \{\mathrm{a},\mathrm{r}\}.
\end{align*}

Theorems~\ref{thm:1} and~\ref{thm:2} below respectively provide an upper bound and a lower bound on $p$, as a function of $I_1, I_2, \ta, \tr$. In particular, the expressions for two different regimes ($\ta \ne \tr$ and $\ta = \tr$) are different.
\begin{theorem} \label{thm:1}
	(a) Consider the regime in which $\ta \ne \tr$. For any $\epsilon > 0$, if
	\begin{align}
	p &\ge \max\bigg\{\frac{(2(1+\epsilon) - I_{1}) \log n}{(\taaa+\trrr)m}, \frac{(1+\epsilon)\log m}{\min\{\taaa,\trrr\}   n}, \notag \\	
	 &\qquad\qquad\qquad \ \qquad\qquad\qquad\frac{(2(1+\epsilon) - I_{2})\log m}{2\tar n} \bigg\}, \label{eq:achievability1}
	\end{align}
then there exists a sequence of estimators $\Phi= \{\phi_n\}_{n=1}^\infty$ satisfying $\lim_{n \to \infty} P_{\emph{err}}(\phi_n) = 0$.
	
	(b) Consider the regime in which $\ta = \tr$. For any $\epsilon > 0$, if $I_{2} \ge 2(1+\epsilon)$ and
	\begin{align}
	p \!\ge\! \max\!\left\{\frac{(2(1+\epsilon) - I_{1})\log n}{(\taaa+\trrr)m}, \frac{(1+\epsilon)\log m}{\min\{\taaa,\trrr\}   n} \right\}, \label{eq:achievability2}
	\end{align}
	 then there exists  $\Phi\!=\! \{\phi_n\}_{n=1}^\infty$ satisfying $\lim_{n \to \infty} P_{\emph{err}}(\phi_n) \!=\! 0$.
\end{theorem}

Again, the estimator $\Phi$ can be chosen as the ML estimator, and the proof sketch is provided in Section~\ref{sec:thm1}. A few remarks concerning Theorem~\ref{thm:1} are in order. 

\begin{enumerate}[wide, labelwidth=!, labelindent=0pt]
\item Intuitively speaking, the first term in the RHS of~\eqref{eq:achievability1} is the threshold for recovering clusters of users, the second term is the threshold for identifying atypical movies, and the third term is the threshold for recovering clusters of movies. When $\ta = \tr$, recovery of the clusters of movies is guaranteed by requiring the movie graph to satisfy $I_2 > 2$  (see point~4 below), thus the RHS of~\eqref{eq:achievability2} only contains two terms.

\item The difference between $\ta$ and $\tr$ is an important statistic for distinguishing action and romance movies. As the distance between $\ta$ and $\tr$ decreases (i.e., it is harder to distinguish action and romance movies), the third term of~\eqref{eq:achievability1} becomes larger (since $\tar$ decreases correspondingly). This means that it may require more samples for recovery.   

\item When $\ta = \tr$, the expression in~\eqref{eq:achievability1} (which is intended for $\ta \ne \tr$) is invalid since $\tar = 0$; instead, the success criterion is given by Theorem~\ref{thm:1}(b). It is interesting to note that the success criterion in Theorem~\ref{thm:1}(b) can be interpreted as a limiting consequence of~\eqref{eq:achievability1} as $\ta \to \tr$. That is, as $\tar \to 0$, we require $I_2 >2(1+\epsilon)$ so that the third term of~\eqref{eq:achievability1} is non-positive. This yields the success criterion in Theorem~\ref{thm:1}(b). 
On the other hand, when $I_2 < 2$, no achievability result is provided---indeed, Theorem~\ref{thm:2}(b) below shows that exact recovery is impossible. 

\item The reason why $I_2 > 2$ is necessary for $\ta = \tr$ is as follows. Without the movie graph $G_2$, typical action movies and atypical romance movies are statistically indistinguishable, since both of them attract (on average) $\frac{(1-\ta)n}{2}$ men and $\frac{\ta n}{2}$ women. The same is true for atypical action movies and typical romance movies. Thus, $G_2$ is the only piece of information that can be exploited to recover clusters of movies. This leads to the necessity of $I_2 > 2$ as per~\cite{abbe2015exact, mossel2015consistency}. 

In contrast, when $\ta \ne \tr$, the rating information can be exploited (together with $G_2$) to distinguish typical action movies and atypical romance movies, since the former  attracts (on average) $\frac{(1-\ta)n}{2}$ men and $\frac{\ta n}{2}$ women, whereas the latter attracts (on average) $\frac{(1-\tr)n}{2}$ men and $\frac{\tr n}{2}$ women. This is why $I_2 > 2$ is not necessary, as reflected in Theorem~\ref{thm:1}(a). 

\item When both $I_1 > 2$ and $I_2 > 2$, the observation of the rating matrix is still needed for exact recovery; this is in   contrast to Model 1. This is because recovering the nominal rating matrix in Model 2 additionally requires the learner to identify atypical movies, and observing the sub-sampled rating matrix is crucial for identifying atypical movies. 

\item When both graphs are available (i.e., $I_1 > 0$ and $I_2 > 0$), our follow-up work~\cite{zhang2020mc2g} proposes and analyzes a computationally efficient algorithm that works in a sequential manner and meets the information-theoretic achievability bound in Theorem~\ref{thm:1}. Extensive numerical experiments therein also help to validate the correctness and predictive abilities  of Theorem~\ref{thm:1}. 
\end{enumerate}

\begin{theorem}\label{thm:2}
	(a) Consider the regime in which $\ta \ne \tr$. For any $\epsilon > 0$, if 
	\begin{align}
	p &< \max\bigg\{\frac{(2(1-\epsilon) - I_{1})\log n}{(\taaa+\trrr)m}, \frac{(1-\epsilon)\log m}{\min\{\taaa,\trrr\}   n},  \notag \\
	&\qquad\qquad\qquad\qquad\qquad\qquad \ \frac{((1-\epsilon) - I_{2})\log m}{2\tar n} \bigg\}, \label{eq:converse1}
	\end{align}
	then $\lim_{n \to \infty} P_{\emph{err}}(\phi_n) = 1$ for any estimator $\Phi$.
	
	(b) Consider the regime in which $\ta = \tr$. For any $\epsilon > 0$, if $I_{2} < 2(1-\epsilon)$ or
	\begin{align}
	p\! <\! \max\left\{\!\frac{(2(1-\epsilon) - I_{1})\log n}{(\taaa+\trrr)m}, \frac{(1-\epsilon)\log m}{\min\{\taaa,\trrr\}   n} \right\}, \label{eq:converse2}
	\end{align}
	then $\lim_{n \to \infty} P_{\emph{err}}(\phi_n) = 1$ for any estimator $\Phi$.
\end{theorem}

The proof sketch is provided in Section~\ref{sec:thm2}.
When $\ta = \tr$, as we have matching upper and lower bounds, a sharp threshold of $p^*$ is established. When $\ta \ne \tr$, the characterization of  $p^*$ is order optimal---in particular, the upper and lower bounds match exactly for a wide range of parameter space, and match up to a constant factor of two for the remaining parameter space. We discuss the reason for this gap and the challenges involved in the converse proof; see Remark~\ref{remark:challenge}.


The following example considers the case $n = 5m$, and quantifies the benefits of graph side-information by analyzing the achievability bound in Theorem~\ref{thm:1}.

\begin{figure}[t!]
	\begin{subfigure}[t]{0.24\textwidth}
		\includegraphics[width=\textwidth]{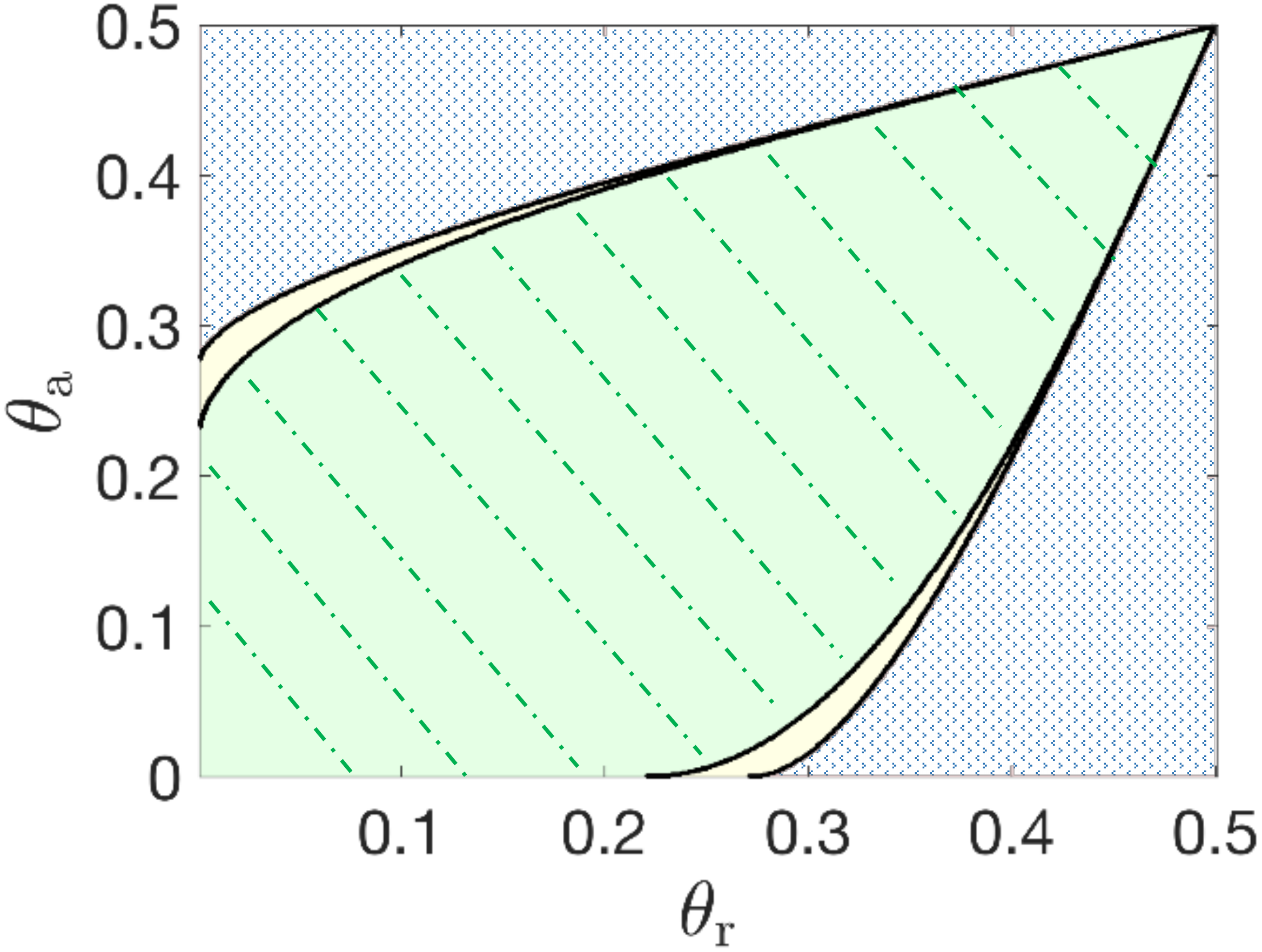}
		\caption{\small The social graph-sensitive (yellow), movie graph-sensitive (green-shaded), and atypicality-sensitive (blue-dotted) regions.}
		\label{fig:3a}
	\end{subfigure}%
	~
	\begin{subfigure}[t]{0.24\textwidth}
		\includegraphics[width=\textwidth]{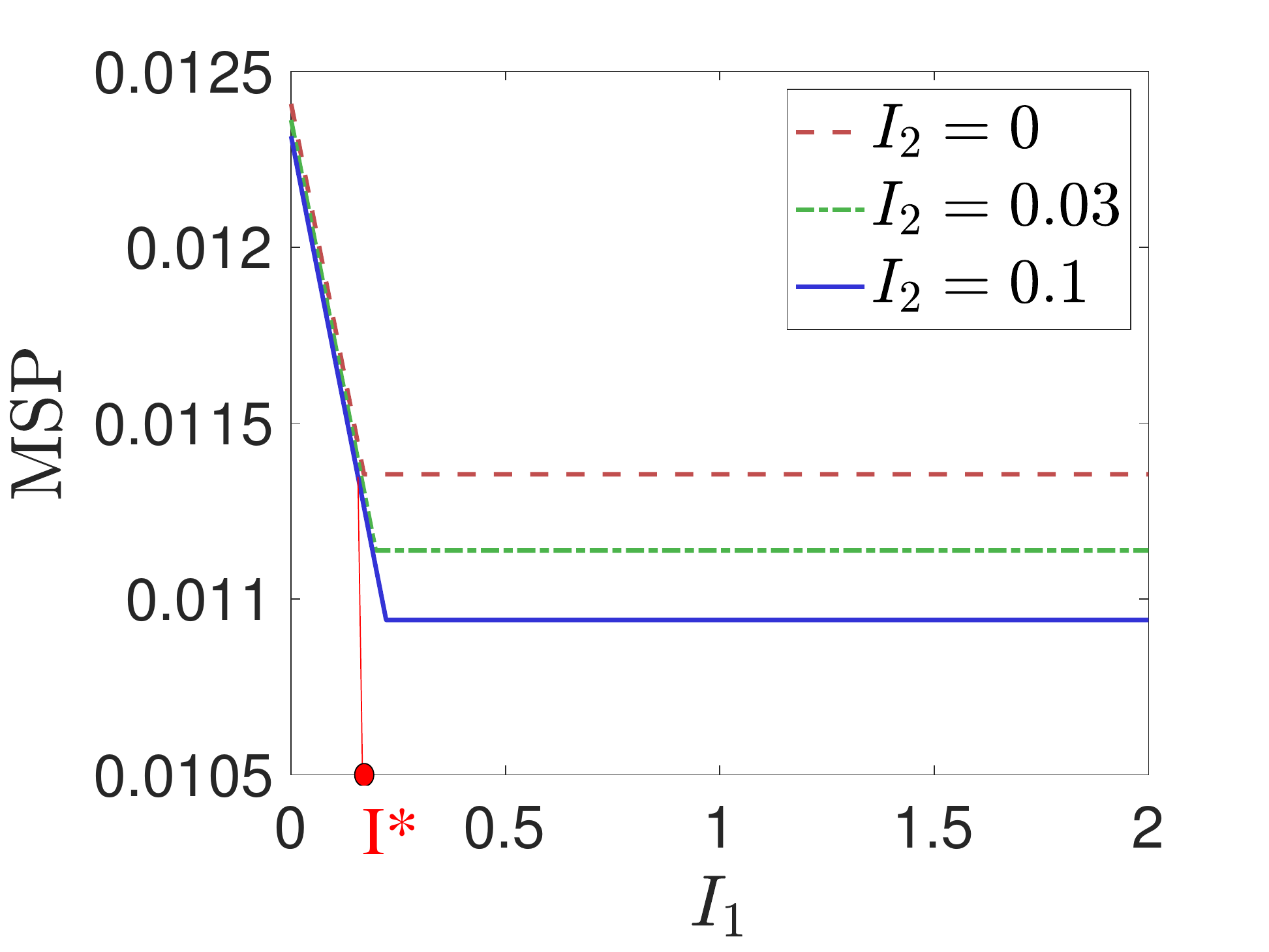}
		\caption{\small The upper bound on the MSP (abbreviated as MSP below) as a function of $I_1$ for $(\ta, \tr) = (0.3, 0.03)$.}
		\label{fig:3c}
	\end{subfigure}%

\begin{subfigure}[t]{0.24\textwidth}
	\includegraphics[width=\textwidth]{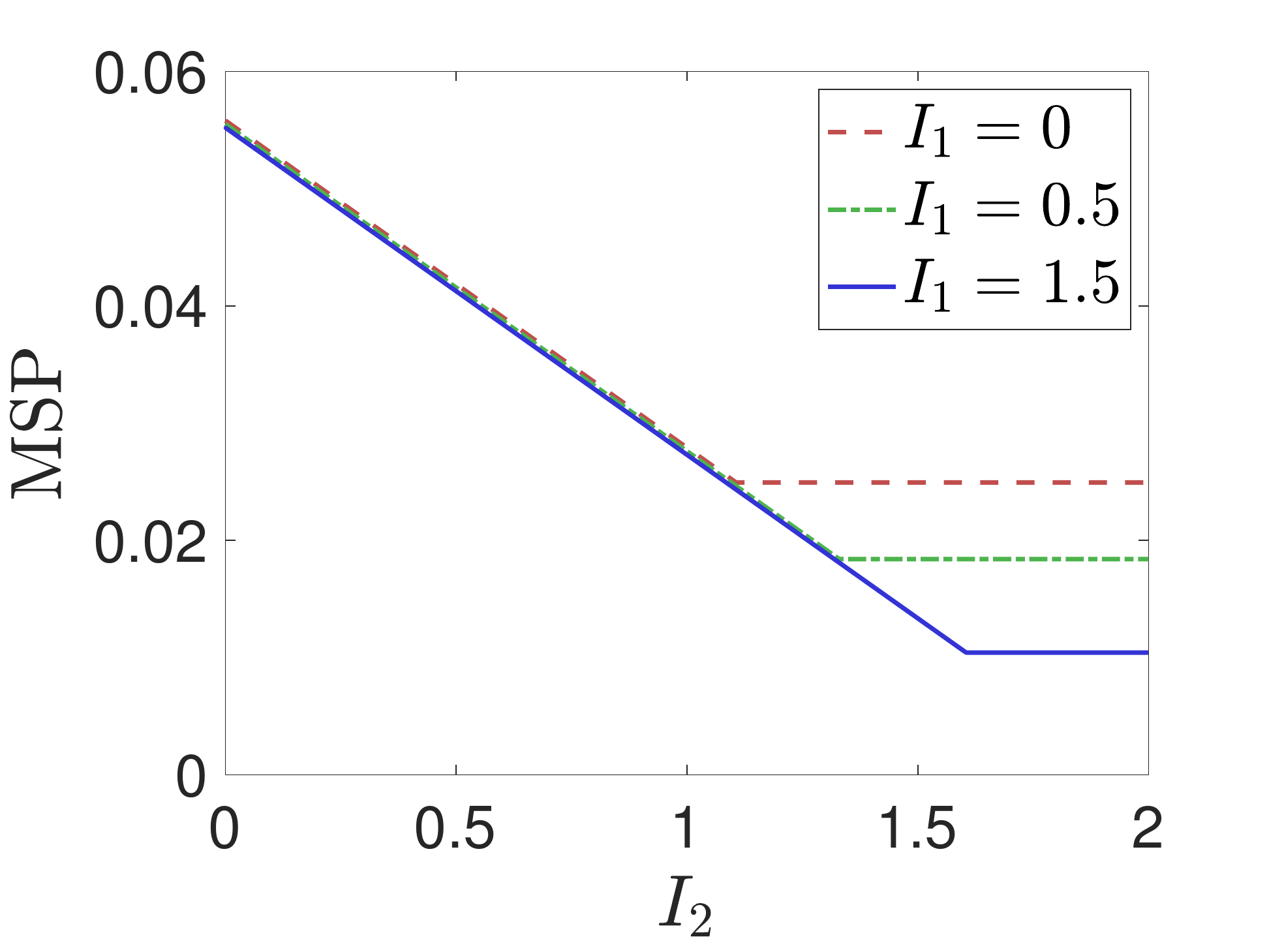}
	\caption{\small  MSP versus $I_2$ for $(\ta, \tr) = (0.3, 0.15)$.}
	\label{fig:3d}
\end{subfigure}%
~
\begin{subfigure}[t]{0.24\textwidth}
	\includegraphics[width=\textwidth]{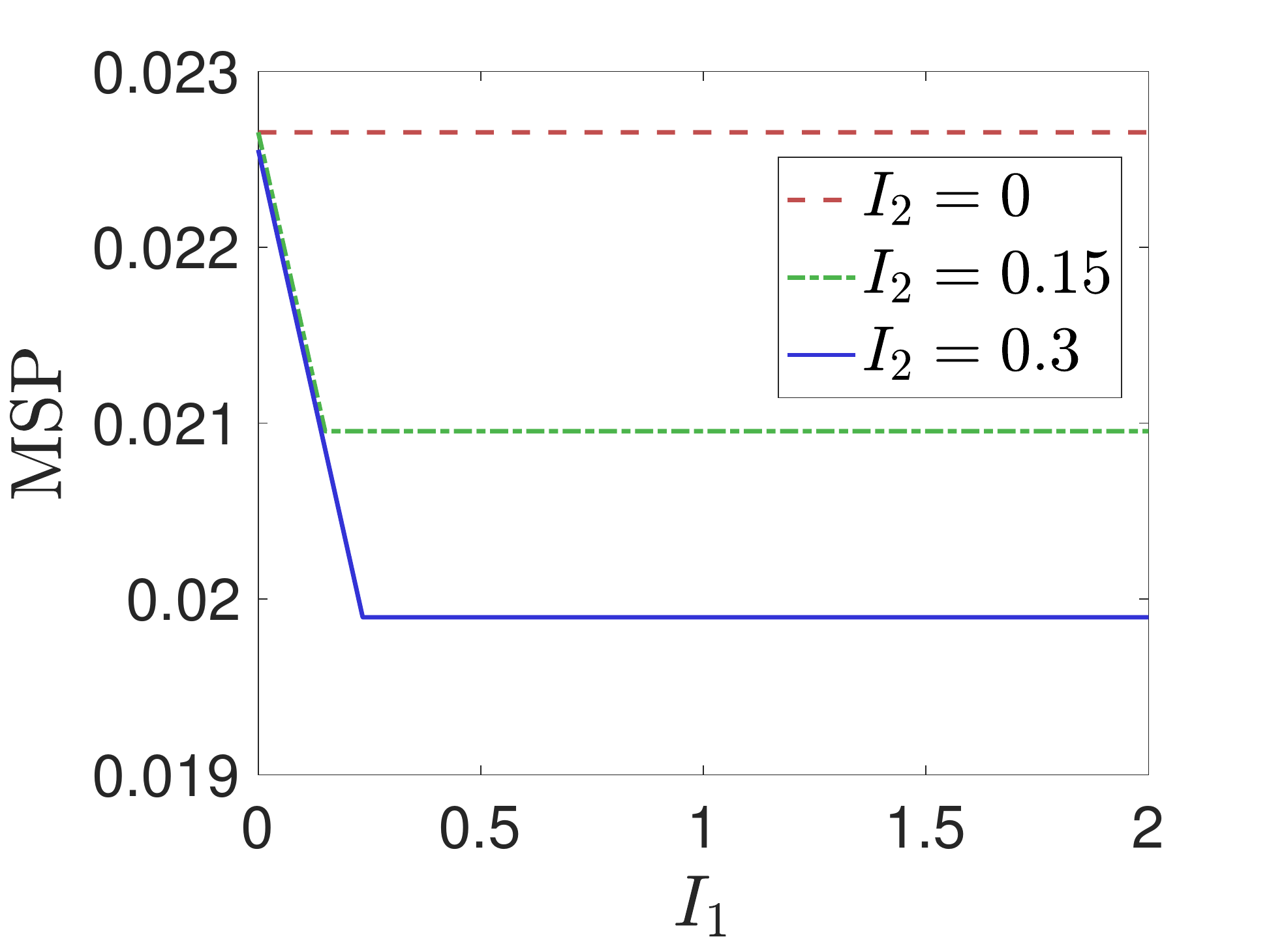}
	\caption{\small MSP versus $I_1$ for $(\ta, \tr) = (0.35, 0.1156)$.}
	\label{fig:3e}
\end{subfigure}
	\caption{The sample complexity gains due to the social and movie graphs for $n = 5m = 10,000$.}
	\vspace{-10pt}
\end{figure}

\begin{example}{\em

First note that the achievability bound in Theorem~\ref{thm:1} depends critically on the values of $(\ta, \tr)$. Fig.~\ref{fig:3a} considers the achievability part in~\eqref{eq:achievability1} and partitions the collections of $(\ta,\tr)$-pairs into three different regimes---the \emph{social graph-sensitive regime} (the yellow region), \emph{movie graph-sensitive regime} (the green-shaded region), and \emph{atypicality-sensitive regime} (the blue-dotted region).

\begin{enumerate}[wide, labelwidth=!, labelindent=0pt]
	
	\item When $(\ta,\tr)$ belongs to the social graph-sensitive regime, observing the social graph $G_1$ is helpful in reducing the \emph{upper bound on the MSP} (which is abbreviated as \emph{MSP} in the following, with a slight abuse of terminology). This is because in this regime and in the absence of graphs (i.e., $I_1 = I_2 = 0$), the MSP is dominated by the first term in~\eqref{eq:achievability1} (i.e., the dominant task is to recover the clusters of users, rather than to   recover the clusters of movies and to identify atypical movies).  Thus, having a positive $I_1$ reduces the MSP. This is reflected in Fig.~\ref{fig:3c}, which plots the MSP as a function of $I_1$ for $(\ta,\tr) = (0.3, 0.03)$ and three different values of~$I_2$. When $I_1 > I^*$ and $I_2 = 0$, the sample complexity gain due to the social graph is ``saturated'', and the movie graph $G_2$ helps to further reduce the MSP. Therefore, observing two graphs (with $I_2 > 0$ and $I_1 > I^*$) is strictly better than observing only one graph; showing a synergistic effect of the two graphs. The reason is similar to that for Fig.~\ref{fig:2}.
	
	\item When $(\ta,\tr)$ belongs to the movie graph-sensitive regime, observing the movie graph $G_2$ is helpful to reduce the MSP. This is reflected in Fig.~\ref{fig:3d}, which plots the MSP as a function of $I_2$ for $(\ta,\tr) = (0.3, 0.15)$ and three different values of $I_1$. This regime is symmetric to the social graph-sensitive regime.

	\item When $(\ta,\tr)$ belongs to the boundary between social graph-sensitive and movie graph-sensitive regimes,  observing both graphs reduces the MSP, while observing only one graph is equivalent to observing neither. The reason for this synergistic effect is similar to that for Fig.~\ref{fig:1}. This observation is also illustrated in Fig.~\ref{fig:3e}, which plots the MSP  as functions of $I_1$ for a boundary point $(\ta,\tr) = (0.35, 0.1156)$.
	

	\item When $(\ta,\tr)$ belongs to the atypicality-sensitive regime, neither the social graph $G_1$ nor the movie graph $G_2$ helps to reduce the MSP. This is because the MSP is dominated by the second term in~\eqref{eq:achievability1} (i.e., more samples are required  to identify atypical movies than to detect clusters of users and movies), thus having positive $I_1$ and/or $I_2$ does not provide any gains.

	\item When $\ta \!=\! \tr$, the movie graph $G_2$ is helpful only if $I_2 > 2$. 
\end{enumerate}

}

\end{example}

\section{Proof of Theorem~\ref{thm:simple}} \label{sec:simple}
We first introduce an important technical lemma.
\begin{lemma}[Chernoff bound]\label{lemma:generic} 
	Consider a random variable $X$. For any $a \in \mathbb{R}$, 
	$$\PP(X>a) \le \min_{t>0} e^{-ta}\cdot \E(e^{tX}).$$ 
\end{lemma}

\subsection{Proof of Achievability}  \label{sec:achievability}
Suppose the sample probability $p$ satisfies~\eqref{eq:achievability}. The proof relies on the ML estimator $\phi_{\text{ML}}$. 
For any $\xi \in \Xi$, the negative log-likelihood of $\xi$ is defined as $L(\xi) \triangleq -\log \PP_{\xi}(\U,G_1,G_2)$. The ML estimator $\phi_{\text{ML}}$ takes $(\U,G_1,G_2)$ as input, and outputs the most likely  instance $\xi$ from the parameter space $\Xi$, i.e.,
\begin{align}
\phi_{\text{ML}}(\U, G_1,G_2) =  \underset{\xi \in \Xi}{\operatorname{argmin}} \ L(\xi). \label{eq:ml}
\end{align}

For any instance $\xi \in \Xi$, the independence of the observations $\U,G_1,G_2$ yields:   
\begin{align*}
\PP_{\xi}(\U,G_1,G_2) = \PP_{\xi}(\U)\PP_{\xi}(G_1)\PP_{\xi}(G_2).
\end{align*} 
Let $\Omega \triangleq \{ (i,j)\in [n] \times [m]: \U_{ij} \ne \perp \}$ be the set of $(i,j)$-pairs such that  $\U_{ij}$ is not erased, and $\Pi_{\xi} \triangleq \{(i,j)\in \Omega: (B_{\xi})_{ij} = \U_{ij} \}$ be the set of $(i,j)$-pairs such that the observed actual rating  equals the nominal rating. Thus, 
\begin{align*}
\PP_{\xi}(\U) = p^{|\Omega|}(1-p)^{mn-|\Omega|} (1-\theta)^{|\Pi_{\xi}|} \theta^{|\Omega|-|\Pi_{\xi}|}.
\end{align*}

Recall that the social graph $G_1 = (\mathcal{V}_1, \mathcal{E}_1)$ is generated as per a binary symmetric SBM. Let $\{X_{ii'} \} \stackrel{\text{i.i.d.}}{\sim} \text{Bern}(\alpha_1)$ and $\{Y_{ii'} \} \stackrel{\text{i.i.d.}}{\sim} \text{Bern}(\beta_1)$. Then, a pair of users  $(i,i') \in \e_1$ if one of the following two is true:
\begin{itemize}
	\item $X_{ii'}=1$, and $(i,i')$ belong to the same cluster.
	\item $Y_{ii'}=1$, and $(i,i')$ belong to different clusters.
\end{itemize}   
Similarly, let $\{\tX_{jj'} \} \stackrel{\text{i.i.d.}}{\sim} \text{Bern}(\alpha_2)$ and $\{\tY_{jj'} \} \stackrel{\text{i.i.d.}}{\sim} \text{Bern}(\beta_2)$. For the movie graph $G_2 = (\mathcal{V}_2, \mathcal{E}_2)$, a pair of movies $(j,j') \in \e_2$ if one of the following two is true:
\begin{itemize}
	\item $\tX_{jj'}=1$, and $(j,j')$ belong to the same cluster.
	\item $\tY_{jj'}=1$, and $(j,j')$ belong to different clusters.
\end{itemize}

We also define the number of edges crossing clusters $\xi_{\mathcal{M}}$ and $\xi_{\mathcal{W}}$ in $G_1$ as $e(\xi_{\mathcal{M}},\xi_{\mathcal{W}})$, and the number of edges crossing  clusters $\xi_{\mathcal{A}}$ and $\xi_{\mathcal{R}}$ in $G_2$ as $e(\xi_{\mathcal{A}},\xi_{\mathcal{R}})$. The probabilities of observing $G_1$ and $G_2$ are respectively
\begin{align*}
&\PP_{\xi}(G_1) = \beta_1^{e(\xi_\mathcal{M},\xi_\mathcal{W})}(1-\beta_1)^{(\frac{n}{2})^2-e(\xi_\mathcal{M},\xi_\mathcal{W})} \\
&\qquad\qquad\times \alpha_1^{|\e_1|-e(\xi_\mathcal{M},\xi_\mathcal{W})}(1-\alpha_1)^{2\binom{n/2}{2}-|\e_1|+e(\xi_\mathcal{M},\xi_\mathcal{W})}, \\
&\PP_{\xi}(G_2) = \beta_2^{e(\xi_\mathcal{A},\xi_\mathcal{R})}(1-\beta_2)^{(\frac{m}{2})^2-e(\xi_\mathcal{A},\xi_\mathcal{R})}\\
&\qquad\qquad\times\alpha_2^{|\e_2|-e(\xi_\mathcal{A},\xi_\mathcal{R})}(1-\alpha_2)^{2\binom{m/2}{2}-|\e_2|+e(\xi_\mathcal{A},\xi_\mathcal{R})}.
\end{align*} 
Let $d_1 \triangleq \log\left(\frac{\alpha_1(1-\beta_1)}{\beta_1(1-\alpha_1)}\right)$, $d_2 \triangleq \log\left(\frac{\alpha_2(1-\beta_2)}{\beta_2(1-\alpha_2)}\right)$, and $f(x) \triangleq \log\left((1-x)/x\right)$.
One can then rewrite $L(\xi)$ as
\begin{align}
L(\xi) 
&= -(\log \PP_{\xi}(\U) + \log \PP_{\xi}(G_1) + \log \PP_{\xi}(G_2)) \notag  \\
&= \!- \lt|\Pi_{\xi}| \!+\! d_1 e(\xi_\mathcal{M},\xi_\mathcal{W}) \!+\! d_2e(\xi_\mathcal{A},\xi_\mathcal{R}) \!+\! C_0 \label{eq:three},
\end{align}
where $C_0$ is a constant that is independent of $\xi$. 

\noindent{\underline{\it Error Analysis:}} Recall from~\eqref{eq:error} that $P_{\text{err}}(\phi_{\text{ML}})$ is defined as the maximum error probability with respect to all possible ground truth parameters $\xi \in \Xi$. We first consider a specific instance $\xi \in \Xi$, and upper-bound the corresponding error probability $\PP_{\xi}(\phi_{\text{ML}}(\U, G_1,G_2) \ne \xi)$. By union bound, we have 
\begin{align}
\PP_{\xi}(\phi_{\text{ML}}(\U, G_1,G_2) \!\ne\! \xi) \!\le\! \sum_{\xi' \in \Xi: \xi' \ne \xi} \PP_{\xi}(L(\xi') \!\le\! L(\xi)). \label{eq:sy2}
\end{align}
In order to analyze $\PP_{\xi}(L(\xi') \le L(\xi))$ for different $\xi' \ne \xi$, we first calculate $L(\xi) - L(\xi')$ based on the expression in~\eqref{eq:three}.

We partition the set of users $\V_1$ into four disjoint subsets 
\begin{align*}
&\Imm \triangleq  \xi_{\mathcal{M}} \cap \xi'_{\mathcal{M}} ,\;  \ \Imw \triangleq \xi_{\mathcal{M}} \cap \xi'_{\mathcal{W}}  , \\
& \Iww \triangleq   \xi_{\mathcal{W}} \cap \xi'_{\mathcal{W}} , \; \ \ \Iwm \triangleq    \xi_{\mathcal{W}} \cap  \xi'_{\mathcal{M}}  .
\end{align*}
Under the ground truth $\xi$, we have 
\begin{align}
e(\xi_\mathcal{M},\xi_\mathcal{W}) &= \sum_{(i,i') \in \xi_{\mathcal{M}} \times \xi_{\mathcal{W}} }  Y_{ii'}, \quad \text{and} \notag \\
e(\xi'_\mathcal{M},\xi'_\mathcal{W}) &= \sum_{i \in \Imm} \sum_{i' \in \Iww} Y_{ii'} + \sum_{i \in \Iwm} \sum_{i' \in \Iww} X_{ii'} \notag  \\
&+\sum_{i \in \Imm} \sum_{i' \in \Imw} X_{ii'} + \sum_{i \in \Iwm} \sum_{i' \in \Imw} Y_{ii'}. \notag
\end{align}
Thus, $e(\xi_\mathcal{M},\xi_\mathcal{W}) - e(\xi'_\mathcal{M},\xi'_\mathcal{W})$ equals
\begin{align}
 &\sum_{(i,i')\in \Imm \times \Iwm}\!\! Y_{ii'} - \sum_{(i,i')\in \Imm \times \Imw}\!\! X_{ii'} \notag \\
 &\quad +\sum_{(i,i')\in \Iww \times \Imw} \!\!Y_{ii'} - \sum_{(i,i')\in \Iww \times \Iwm} \!\!X_{ii'} \triangleq \Gamma_1. \label{eq:fan1}
\end{align}
Note that for any $\xi' \in \Xi$, $|\Imm| = |\Iww|$ and $|\Imw| = |\Iwm|$. We define $k_1 \triangleq |\Imw|$ as the parameter that quantifies the amount of overlap between the communities of men and women in $\xi$ and $\xi'$, where $k_1 \in [0:n/4]$. It can be shown that $\Gamma_1$ in~\eqref{eq:fan1} contains $g_1 \triangleq nk_1-2k_1^2$ copies of $(Y_{ii'} - X_{ii'})$.

Similarly, we partition the set of movies $\V_2$ into subsets 
\begin{align}
&\Iaa \triangleq   \xi_{\mathcal{A}} \cap \xi'_{\mathcal{A}} ,\; \ \Iar \triangleq  \xi_{\mathcal{A}} \cap \xi'_{\mathcal{R}} , \label{eq:i1} \\
& \Irr \triangleq  \xi_{\mathcal{R}} \cap \xi'_{\mathcal{R}} ,\; \ \Ira \triangleq  \xi_{\mathcal{R}} \cap \xi'_{\mathcal{A}} , \label{eq:i2}
\end{align}
and one can show that $e(\xi_\mathcal{A},\xi_\mathcal{R}) - e(\xi'_\mathcal{A},\xi'_\mathcal{R})$ equals
\begin{align}
&\sum_{(j,j')\in \Iaa \times \Ira}\!\! \tY_{jj'} - \sum_{(j,j')\in \Iaa \times \Iar}\!\! \tX_{jj'} \notag \\
&\quad +\sum_{(j,j')\in \Irr \times \Iar} \!\!\tY_{jj'} - \sum_{(j,j')\in \Irr \times \Ira} \!\!\tX_{jj'} \triangleq \Gamma_2. \label{eq:fan2}
\end{align}
Let $k_2 \triangleq |\Iar|$ be the parameter that quantifies the amount of overlap between the communities of action and romance movies in $\xi$ and $\xi'$, where $k_2 \in [0:m/4]$. It can be shown that $\Gamma_2$ in~\eqref{eq:fan2} contains $g_2 \triangleq mk_2-2k_2^2$ copies of $(\tY_{jj'} - \tX_{jj'})$. 

Let $\s \triangleq \{(i,j)\in [n] \times [m]: \bbij = \bij\}$ be the set of $(i,j)$-pairs such that the $(i,j)$-th entry in $B_{\xi}$ and the $(i,j)$-th entry in $B_{\xi'}$ coincide. We then have
\begin{align*}
&|\Pi_{\xi}| = \! \!\sum_{(i,j) \in \s} \mathbbm{1}\left\{ \vij = \bbij \right\} + \! \!\sum_{(i,j) \in \s^c} \mathbbm{1}\left\{ \vij = \bbij \right\}, \\
&|\Pi_{\xi'}| = \! \!\sum_{(i,j) \in \s} \!\!\mathbbm{1}\left\{ \vij = \bij \right\} + \! \!\sum_{(i,j) \in \s^c}\!\!\! \mathbbm{1}\left\{ \vij = \bij \right\} \\
&\quad \ \ =\! \!\sum_{(i,j) \in \s} \!\!\mathbbm{1}\!\left\{ \vij = \bbij \right\} + \! \!\sum_{(i,j) \in \s^c} \!\!\!\mathbbm{1}\!\left\{ \vij = 1-\bbij \right\}.
\end{align*}
By noting that $\mathbbm{1}\left\{ \vij = \bbij \right\} = \zij(1-\tij)$, we obtain
\begin{align}
|\Pi_{\xi}| - |\Pi_{\xi'}| 
&= \sum_{(i,j) \in \s^c} \zij (1-2\tij) \triangleq \Gamma_3, \label{eq:fine}
\end{align}
where  $|\s^c| = g_3 \triangleq 2mk_1 + 2nk_2 - 8k_1k_2$. Thus, $\Gamma_3$ in~\eqref{eq:fine} contains $g_3$ copies of $\zij (1-2\tij)$. 

Combining~\eqref{eq:three},~\eqref{eq:fan1},~\eqref{eq:fan2}, and~\eqref{eq:fine}, we have 
\begin{align}
	L(\xi) - L(\xi') = d_1 \Gamma_1 + d_2 \Gamma_2 -\lt \Gamma_3. \label{eq:18}	
\end{align}
By applying  Chernoff bound  (Lemma~\ref{lemma:generic}) 
with $t = \frac{1}{2}$, we have  
\begin{align}
&\PP_{\xi}\big(L(\xi')\! \le\! L(\xi)\big) \!\le \! \exp \! \Big\{\! \!\!-g_1I_{1}\frac{\log n}{n}
 \!-\! g_2I_{2}\frac{\log m}{m} \!- \!g_3 p\ttt \!\Big\}. \label{eq:mo}
\end{align}
Note that the expression above depends on $k_1$ and $k_2$, since $g_1,g_2,g_3$ are all functions of $k_1$ and $k_2$. Let $\Xi_{\xi}(k_1,k_2)$ be the set of $\xi'$ with the same error probability $\PP_{\xi}(L(\xi') \le L(\xi))$, and $\mathcal{T}$ be the set of valid $(k_1,k_2)$-pairs, i.e.,
\begin{align}
&\mathcal{T} \!\triangleq \! \left\{(k_1,k_2) \!\ne\! (0,0): k_1 \in \left[0:\frac{n}{4}\right], k_2 \in \left[0:\frac{m}{4}\right] \right\}. \label{eq:ni}
\end{align}
Following~\eqref{eq:sy2}, we then decompose the parameter space $\Xi$ into different \emph{type classes} $\Xi_{\xi}(k_1,k_2)$:
\begin{align}
&\sum_{\xi' \in \Xi: \xi' \ne \xi} \PP_{\xi}(L(\xi') \le L(\xi)) \notag \\
&=  \sum_{(k_1,k_2) \in \mathcal{T}} |\Xi_{\xi}(k_1,k_2)| \cdot  \PP_{\xi}(L(\xi') \le L(\xi)). \label{eq:type}
\end{align}
For any $\epsilon > 0$, we define an auxiliary parameter
$\delta_{\epsilon} \triangleq \epsilon/(4+4\epsilon)$. We also define four subsets of $\mathcal{T}$ as 
\begin{align}
&\mathcal{T}_1 \triangleq \left\{(k_1,k_2)\!: k_1 \in [1:\delta_{\epsilon} n], k_2 \in [0:\delta_{\epsilon} m]\right\}, \label{eq:t1} \\
&\mathcal{T}_2 \triangleq \left\{(k_1,k_2)\!: k_1 = 0,  k_2 \in [1:\delta_{\epsilon} m] \right\}, \label{eq:t2}\\
&\mathcal{T}_3 \triangleq \left\{(k_1,k_2)\!: k_1 \in [\delta_{\epsilon} n+1:n/4], k_2 \in \left[0:m/4\right] \right\},\label{eq:t3} \\
&\mathcal{T}_4 \triangleq \left\{(k_1,k_2)\!: k_1 \in \left[0\!:\!n/4\right], k_2 \in  [\delta_{\epsilon} m+1 : m/4] \right\}. \label{eq:t4}
\end{align}
Note that $\mathcal{T} \subset \cup_{i=1}^4 \mathcal{T}_i$. Thus,~\eqref{eq:type} is upper-bounded by
\begin{align}
\sum_{(k_1,k_2) \in \left(\cup_{i=1}^4 \mathcal{T}_i\right)} |\Xi_{\xi}(k_1,k_2)| \cdot  \PP_{\xi}(L(\xi') \le L(\xi)). \label{eq:c}
\end{align}
We now show that the error probabilities induced by each of the four subsets $\mathcal{T}_1$, $\mathcal{T}_2$, $\mathcal{T}_3$, $\mathcal{T}_4$ vanish as $n$ tends to infinity.

\subsubsection{Case 1: $k_1 \le \delta_{\epsilon} n$ and $k_2 \le \delta_{\epsilon} m$}
In this case, $g_1 = k_1\left(n-2k_1\right) \ge k_1n\left(1-2\delta_{\epsilon}\right)$ and $g_2 \ge k_2m\left(1-2\delta_{\epsilon}\right)$. Given $(k_1,k_2)$, one can bound the RHS of~\eqref{eq:mo} from above by  
\begin{align}
&\exp\!\big\{\!\! -\! k_1\! \left(1\!-\!2\delta_{\epsilon}\right)I_{1}(\log n) \!-\! k_2\left(1\!-\!2\delta_{\epsilon}\right)I_{2}(\log m) \!-\! g_3p \ttt  \big\} \notag \\
&= \exp\big\{\!\!-\!k_1\big[\!\left(1-2\delta_{\epsilon}\right)I_{1}(\log n)+(2m-4k_2)p\ttt \big] \big\}  \notag \\
&\qquad \times \exp\big\{\!\!-\!k_2\big[\!\left(1-2\delta_{\epsilon}\right)I_{2}(\log m)+(2n-4k_1)p\ttt \big] \big\} \notag \\
&\le \exp\big\{\!\!-\!k_1\left(1-2\delta_{\epsilon}\right)\big[I_{1}(\log n)+2mp\ttt \big] \big\}  \notag \\
&\qquad  \times \exp\big\{\!\!-\!k_2\left(1-2\delta_{\epsilon}\right)\big[I_{2}(\log m)+2np\ttt \big] \big\}. \label{eq:jay}
\end{align}
Since the sample probability $p$ satisfies~\eqref{eq:achievability}, we have 
\begin{align}
&k_1\left(1-2\delta_{\epsilon}\right)\big[I_{1}(\log n)+2mp\ttt \big]
\!\ge\! k_1\left(2+\epsilon\right) \log n, \label{eq:com1}\\
&k_2\left(1-2\delta_{\epsilon}\right)\big[I_{2}(\log m)+2np\ttt \big]
\!\ge\! k_2\left(2+\epsilon\right) \log m.\label{eq:com2}
\end{align}
Combining~\eqref{eq:jay},~\eqref{eq:com1}, and~\eqref{eq:com2}, we get 
\begin{align}
\PP_{\xi}(L(\xi') \le L(\xi)) \le \exp\left\{\left(2+\epsilon\right)(k_1\log n + k_2 \log m) \right\}. \notag 
\end{align}
Also note that the size of $\Xi_{\xi}(k_1,k_2)$ satisfies
\begin{align}
\binom{\frac{n}{2}}{k_1}\binom{\frac{n}{2}}{k_1}\binom{\frac{m}{2}}{k_2}\binom{\frac{m}{2}}{k_2} \le n^{2k_1}\cdot m^{2k_2}. \label{eq:com3}
\end{align}
Thus, the error probability induced by $\mathcal{T}_1$ is upper-bounded as
\begin{align*}
&\sum_{(k_1,k_2) \in \mathcal{T}_1} |\Xi_{\xi}(k_1,k_2)| \cdot  \PP_{\xi}(L(\xi') \le L(\xi)) \\
&\le \sum_{(k_1,k_2) \in \mathcal{T}_1}\!\! n^{2k_1} m^{2k_2} \cdot \exp\left\{\left(2+\epsilon\right)(k_1\log n + k_2 \log m) \right\} \\
&\le \sum_{k_1 \in [1:\delta_{\epsilon} n]} n^{-\epsilon k_1} \sum_{k_2 \in [0:\delta_{\epsilon} m]} m^{-\epsilon k_2}  \\
&\le 4 n^{-\epsilon}.
\end{align*}
Similarly, the error probability induced by $\mathcal{T}_2$ satisfies 
\begin{align*}
\sum_{(k_1,k_2) \in \mathcal{T}_2} |\Xi_{\xi}(k_1,k_2)| \cdot  \PP_{\xi}(L(\xi') \le L(\xi)) \le 2 m^{-\epsilon}.
\end{align*}

\subsubsection{Case 2: $k_1 > \delta_{\epsilon} n$ or $k_2 > \delta_{\epsilon} m$}
When $k_1 > \delta_{\epsilon} n$,
$$g_3 \ge 2\delta_{\epsilon} mn + 2nk_2 - 8k_1k_2 \ge 2\delta_{\epsilon} mn,$$
where the last step holds since $k_1 \le n/4$. Similarly, when $k_2 > \delta_{\epsilon} m$, we have $g_3 \ge 2\delta_{\epsilon} mn$ as well. Thus, in both cases,
$$\PP_{\xi}(L(\xi') \le L(\xi)) \le \exp\left\{-\Omega(pmn)\right\} = e^{-\Omega(m\log m + n \log n)}.$$
Since the number of partitions of the set of $n$ users into $n/2$ men and $n/2$ women is upper-bounded by $2^n$, the number of partitions of the set of $m$ movies into $m/2$ action movies and $m/2$ romance movies is upper-bounded by $2^m$, the size of the parameter space $\Xi$ is at most  $2^{n + m}$. 
Therefore, one can show that the error probability induced by $\mathcal{T}_3 \cup \mathcal{T}_4$ satisfies
\begin{align}
&\sum_{(k_1,k_2) \in \mathcal{T}_3 \cup \mathcal{T}_4} |\Xi_{\xi}(k_1,k_2)| \cdot  \PP_{\xi}(L(\xi') \le L(\xi)) \notag \\
&\qquad\qquad\le 2^{n + m}\cdot e^{-\Omega(m\log m + n \log n)} = o(1). \label{eq:case2}
\end{align}
Combining~\eqref{eq:sy2},~\eqref{eq:type},~\eqref{eq:c}, and the analyses in Case 1 and Case 2, we obtain that $\PP_{\xi}(\phi_{\text{ML}}(\U, G_1,G_2) \!\ne\! \xi) = o(1)$. 

Finally, it is worth noting that the error probability bound derived above is valid regardless of which $\xi \in \Xi$ is set to be the ground truth. Thus, we conclude $$P_{\text{err}}(\phi_{\text{ML}}) = \max_{\xi \in \Xi} \PP_{\xi}(\phi_{\text{ML}}(\U, G_1,G_2) \!\ne\! \xi) = o(1).$$

\subsection{Proof of Converse} \label{sec:converse}

In this subsection, we show that for any $\epsilon > 0$, the error probability $P_{\text{err}}\to 1$ as $n\to\infty$ if $p$ satisfies~\eqref{eq:converse}. First, we state a technical lemma concerning the tightness of the Chernoff bound on the exponential scale.
\begin{lemma}[Adapted from Lemma 4 of \cite{ahnsupp}] \label{lemma:chernoff1}
	Consider integers $K,L_1,L_2 \in \mathbb{N}$ and $\theta_1,\theta_2 \in [0,1]$. Let $\{Y_i\}_{i = 1}^K \stackrel{\mathrm{i.i.d.}}{\sim} \emph{Bern}(\beta_1)$, $\{X_i\}_{i = 1}^K \stackrel{\mathrm{i.i.d.}}{\sim}$ $\emph{Bern}(\alpha_1)$, $\{Z_j^{(1)}\}_{j = 1}^{L_1}, \{Z_j^{(2)}\}_{j = 1}^{L_2} \stackrel{\mathrm{i.i.d.}}{\sim} \emph{Bern}(p)$, $\{\Theta_j^{(1)}\}_{j = 1}^{L_1} \stackrel{\mathrm{i.i.d.}}{\sim} \emph{Bern}(\theta_2)$, $\{\Theta^{(2)}_j\}_{j = 1}^{L_2} \stackrel{\mathrm{i.i.d.}}{\sim} \emph{Bern}(\theta_2)$, and assume that $\alpha_1, \beta_1, p = o(1)$ and $\max\{\sqrt{\alpha_1\beta_1}K, pL_1, pL_2\} = \omega(1)$. Then, 
	\begin{align*}
	&\PP\bigg(g_1\sum_{i=1}^K (Y_i-X_i) + \sum_{k=1}^2 f(\theta_k) \sum_{j=1}^{L_k} Z^{(k)}_j(2\Theta^{(k)}_j-1) \ge 0 \bigg) \\
	& \ge  \frac{1}{4}\exp\bigg\{\!-(1\!+\!o(1))KI_1\frac{\log n}{n} -\sum_{k=1}^2 (1\!+\!o(1))p L_k h(\theta_k) \bigg\}.
	\end{align*} 
\end{lemma}

Following the definition of $P_{\text{err}}(\phi)$, we have
\begin{align}
\inf_{\phi} P_{\text{err}}(\phi) &\ge \inf_{\phi} \frac{1}{|\Xi|}\sum_{\xi \in \Xi} \PP_{\xi}(\phi(\U, G_1,G_2) \ne \xi) \notag \\
& = \frac{1}{|\Xi|}\sum_{\xi \in \Xi} \PP_{\xi}(\phi_{\mathrm{ML}}(\U, G_1,G_2) \ne \xi)  \label{eq:sj1},
\end{align}
where~\eqref{eq:sj1} holds since the ML estimator is optimal when the prior is uniform. In the following, we analyze the error probability with respect to $\phi_{\text{ML}}$ and a specific ground truth $\xi \in \Xi$.
The crux of the proof is to focus on a subset of events (corresponding to a particular type class $\Xi_{\xi}(k_1,k_2)$) that are most likely  to cause errors, and show that the  error probability tends to one even if we restrict the error analysis to this type class $\Xi_{\xi}(k_1,k_2)$. 

Note that the ML estimator $\phi_{\mathrm{ML}}$ succeeds if $L(\xi') > L(\xi)$ holds for all $\xi \ne \xi'$. Thus, the success probability $P_{\text{suc}}(\xi)$, which equals $1 - \PP_{\xi}(\phi_{\mathrm{ML}}(\U, G_1,G_2) \ne \xi)$, takes the form $\PP_{\xi}(\bigcap_{\xi'\ne \xi} \{L(\xi') \!>\! L(\xi)\})$. Clearly, we also have
\begin{align}
	&P_{\text{suc}}(\xi) \le \PP_{\xi}\Bigg(\bigcap_{\xi' \in \Xi_{\xi}(k_1,k_2)} \{L(\xi') \!>\! L(\xi)\}\Bigg), \label{eq:sub}
\end{align}
where $\Xi_{\xi}(k_1,k_2)$ can be chosen as any type class such that $(k_1, k_2) \in \mathcal{T}$.
Lemma~\ref{lemma:10} below focuses on the type classes $\Xi_{\xi}(1,0)$ and $\Xi_{\xi}(0,1)$, and shows that the success probability tends to zero if $p$ satisfies~\eqref{eq:converse}.    

\begin{lemma}\label{lemma:10}
Consider any ground truth $\xi \in \Xi$ and sufficiently large $n$. (i) When $p \le \frac{(2(1-\epsilon) - I_{1}) \log n}{2\ttt m}$, we have
\begin{align}
\PP_{\xi}\Bigg(\bigcap_{\xi' \in \Xi_{\xi}(1,0)} \{L(\xi') \!>\! L(\xi)\}\Bigg) \le 5\exp\Big(-\frac{1}{4}n^{\frac{\epsilon}{2}}\Big).\label{eq:first}
\end{align}
(ii) When $p \le \frac{(2(1-\epsilon) - I_{2}) \log m}{2\ttt n}$, we have
\begin{align}
\PP_{\xi}\Bigg(\bigcap_{\xi' \in \Xi_{\xi}(0,1)} \{L(\xi') \!>\! L(\xi)\}\Bigg) \le 5\exp\Big(-\frac{1}{4}m^{\frac{\epsilon}{2}}\Big). \label{eq:second}
\end{align}
\end{lemma}

The converse part of Theorem~\ref{thm:simple} follows directly from~\eqref{eq:sj1},~\eqref{eq:sub}, and Lemma~\ref{lemma:10}.
In the following, we provide a detailed proof of Lemma~\ref{lemma:10}.

\begin{proof}[Proof of Lemma~\ref{lemma:10}]
Suppose the ground truth $\xi$ satisfies $\xi_{\mathcal{M}} = [1:n/2]$, $\xi_{\mathcal{W}} = [(n/2)+1:n]$. The instantiation of $\xi$ is merely for ease of presentation, and it will be clear that the following proof is valid for every $\xi \in \Xi$.  

Let $r_1 \triangleq n/\log^2 n$, $\mathcal{V}'_1 \triangleq [1:2r_1] \cup [(n/2)+1: (n/2)+2r_1] \subseteq \mathcal{V}_1$ be a subset of user nodes, and $\Delta_1$ be the event that the number of isolated nodes in $\mathcal{V}'_1$ (i.e., the nodes that are not connected to any other nodes in $\mathcal{V}'_1$) is at least $3r_1$.
Lemma~\ref{lemma:strong} below shows that $\Delta_1$ occurs with high probability over the generation of the social graph $G_1$.
\begin{lemma} \label{lemma:strong}
	The probability  that event $\Delta_1$ occurs is at least $1 - \exp\left(-\frac{\eta^2(\alpha_1+\beta_1)n^2}{\log^4 n}\right)$, where $\eta \in (0,1)$ is arbitrary.
\end{lemma}

For each $i \in \xi_{\mathcal{M}}$ and $i' \in \xi_{\mathcal{W}}$, we define variants of $\xi$ as follows:
\begin{enumerate}
	\item $\xi_{\text{row}}^{(i)}$ is identical to $\xi$ except that $(\xi_{\text{row}}^{(i)})_{\mathcal{M}} = \xi_{\mathcal{M}} \setminus \{i\}$ and $(\xi_{\text{row}}^{(i)})_{\mathcal{W}} = \xi_{\mathcal{W}} \cup \{i\}$, i.e., user $i$ in $\xi_{\text{row}}^{(i)}$ is female;
	
	\item  $\xi_{\text{row}}^{(i')}$ is identical to $\xi$ except that $(\xi_{\text{row}}^{(i')})_{\mathcal{M}} = \xi_{\mathcal{M}} \cup \{i'\}$ and $(\xi_{\text{row}}^{(i')})_{\mathcal{W}} = \xi_{\mathcal{W}} \setminus \{i'\}$, i.e., user $i'$ in $\xi_{\text{row}}^{(i)}$ is male;
	\item $\xi_{\text{row}}^{(i,i')}$ is identical to $\xi$ except that user $i$ and user $i'$ in $\xi_{\text{row}}^{(i,i')}$ are respectively female and male.
\end{enumerate}
By definition, the type class $\Xi_{\xi}(1,0)$ is equivalent to the set  $\{\xi_{\text{row}}^{(i,i')}: i \in \xi_{\mathcal{M}}, i' \in \xi_{\mathcal{W}} \}$.
Conditioned on $\Delta_1$, one can find a subset $\xi_{P_1} \subset \xi_\mathcal{M}$ and a subset $\xi_{Q_1} \subset \xi_\mathcal{W}$ such that $|\xi_{P_1}| = |\xi_{Q_1}| = r_1$ and all the nodes in $\xi_{P_1} \cup \xi_{Q_1}$ are not connected to one another. Thus,
\begin{align}
&\PP_{\xi}\Bigg(\bigcap_{\xi' \in \Xi_{\xi}(1,0)} \{L(\xi') > L(\xi)\}  \Bigg) \notag \\
&\le \PP_{\xi}\Bigg(\bigcap_{i \in \xi_{P_1}, i' \in \xi_{Q_1}} \left\{ L\left(\xi_{\text{row}}^{(i,i')}\right) > L(\xi) \right\} \Bigg) \notag \\
&= \PP_{\xi}(\Delta_1) \PP\Bigg(\bigcap_{i \in \xi_{P_1}, i' \in \xi_{Q_1}} \left\{L\left(\xi_{\text{row}}^{(i,i')}\right) > L(\xi) \right\} \bigg| \Delta_1 \Bigg) \notag \\
&\quad + \PP_{\xi}(\Delta_1^c) \PP\Bigg(\bigcap_{i \in \xi_{P_1}, i' \in \xi_{Q_1}} \left\{ L\left(\xi_{\text{row}}^{(i,i')}\right) > L(\xi) \right\} \bigg| \Delta_1^c \Bigg) \notag\\
&\le \PP_{\xi}\Bigg(\!\bigcap_{i \in \xi_{P_1}, i' \in \xi_{P_1}} \!\!\left\{ L\left(\xi_{\text{row}}^{(i,i')}\right) > L(\xi)\right\} \bigg| \Delta_1 \Bigg) + \PP_{\xi}(\Delta_1^c) \notag \\
&\le \PP_{\xi}\Bigg(\bigcap_{i \in \xi_{P_1}} \left\{ L(\xi_{\text{row}}^{(i)}) > L(\xi) \right\} \bigg| \Delta_1 \Bigg) \notag \\
&+ \PP_{\xi}\left(\bigcap_{i' \in \xi_{Q_1}} \left\{L(\xi_{\text{row}}^{(i')}) > L(\xi) \right\} \bigg| \Delta_1 \right) + \PP_{\xi}(\Delta_1^c), \label{eq:41}
\end{align}
where~\eqref{eq:41} is due to Lemma~\ref{lemma:compose} below, which is borrowed from~\cite[Lemma 6]{ahnsupp}.  
\begin{lemma} \label{lemma:compose}
	Conditioned on $\xi$ and $\Delta_1$, if $L(\xi_{\mathrm{row}}^{(i)}) \le L(\xi)$ and $L(\xi_{\mathrm{row}}^{(i')}) \le L(\xi)$ for some $i \in \xi_{P_1}$ and $i' \in \xi_{Q_1}$, then we have $L(\xi_{\mathrm{row}}^{(i,i')}) \le L(\xi)$.
\end{lemma}
Without loss of generality, we assume $1 \in \xi_{P_1}$ and $(n/2)+1 \in \xi_{Q_1}$. It is worth noting that conditioned on $\Delta_1$ and $\xi$, the events $\{L(\xi_{\text{row}}^{(i)}) > L(\xi)\}$ for different $i \in \xi_{P_1}$ are mutually independent, thus the first term of~\eqref{eq:41} equals
\begin{align} \PP_{\xi}\!\left(\!L(\xi_{\text{row}}^{(1)}) \!>\! L(\xi) \big| \Delta_1 \!\right)^{\!|\xi_{P_1}|}\!\!\! =\! \PP_{\xi}\!\left(\!L(\xi_{\text{row}}^{(1)}) \!>\! L(\xi) \big| \Delta_1 \!\right)^{\!r_1}. \label{eq:mul1}
\end{align}
Similarly, the events $\{L(\xi_{\text{row}}^{(i')}) > L(\xi)\}$ for different $i' \in \xi_{Q_1}$ are mutually independent, thus the second term of~\eqref{eq:41} equals
\begin{align}
 \PP_{\xi}\left(L(\xi_{\text{row}}^{(\frac{n}{2}+1)}) > L(\xi) \big| \Delta_1 \right)^{r_1}. \label{eq:mul2}
\end{align}

\begin{remark}{\em
	The main purpose of introducing $\xi_{P_1}$ and $\xi_{Q_1}$ is to ensure that the events $\{L(\xi_{\text{row}}^{(i)}) > L(\xi)\}_{i \in \xi_{P_1}}$ are mutually independent and the events $\{L(\xi_{\text{row}}^{(i')}) > L(\xi)\}_{i' \in \xi_{Q_1}}$ are mutually independent. }
\end{remark}

By noting that $\PP_{\xi}(\Delta_1) \ge 1 - \exp\left(-\frac{\eta^2(\alpha_1+\beta_1)n^2}{\log^4 n}\right)$ and
\begin{align*}
\PP_{\xi}\!\left(\!L(\xi_{\text{row}}^{(1)}) > L(\xi) \right) &\ge \PP_{\xi}\!\left(L(\xi_{\text{row}}^{(1)}) > L(\xi) \big| \Delta_1 \!\right) \PP_{\xi}(\Delta_1),
\end{align*}
we have 
\begin{align}
\PP_{\xi}\!\left(L(\xi_{\text{row}}^{(1)}) > L(\xi) \big| \Delta_1 \!\right) \le \frac{\PP_{\xi}\!\left(L(\xi_{\text{row}}^{(1)}) > L(\xi)\!\right)}{1\! -\! \exp\!\left(\!-\frac{\eta^2(\alpha_1+\beta_1)n^2}{\log^4 n}\right)}.  \label{eq:xing1}
\end{align}
It then remains to provide an upper bound on $\PP_{\xi}(L(\xi_{\text{row}}^{(1)}) > L(\xi))$. By substituting $\xi'$ with $\xi_{\text{row}}^{(1)}$ in~\eqref{eq:18}, one can formulate $L(\xi) - L(\xi_{\text{row}}^{(1)})$ as per~\eqref{eq:18}, wherein  $\Gamma_1, \Gamma_2, \Gamma_3$ are represented in terms of  $\s^c = \{(1,j)\}_{j=1}^m$ and
\begin{align*}
&\Imm = \Big[2:\frac{n}{2}\Big], \ \Iww = \Big[\frac{n}{2}+1:n\Big], \ \Imw = \emptyset, \ \Iwm = \{1\}.
\end{align*}
We apply Lemma~\ref{lemma:chernoff1} (with $K= (n/2)-1$, $L_1 = m$, $\theta_1 = \theta$, $L_2 = 0$) to bound $\PP_{\xi}(L(\xi_{\text{row}}^{(1)}) \le L(\xi))$ from below by 
\begin{align*}
	\frac{1}{4}\exp\left\{-(1+o(1))\left(\frac{n}{2}-1\right)I_1\frac{\log n}{n} - (1+o(1))mp \ttt \right\}.
\end{align*}
For sufficiently large $n$, we have
\begin{align}
&\PP_{\xi}\left(L(\xi_{\text{row}}^{(1)}) > L(\xi)  \right)^{r_1} \notag \\
&= \left(1 - \PP_{\xi}\left(L(\xi_{\text{row}}^{(1)}) \le L(\xi)  \right) \right)^{r_1} \notag \\ 
& \le  \exp\Big\{-\frac{1}{4} e^{-(1+o(1))\left(\frac{n}{2}-1\right)I_{1}\frac{\log n}{n} -(1+o(1)) mp\ttt} \Big\}^{r_1} \notag \\
&\le \exp(-n^{\frac{\epsilon}{2}}/4), \label{eq:xing2}
\end{align}
where the last step holds since $p \le  \frac{(2(1-\epsilon) - I_{1}) \log n}{2\ttt m}$. Combining~\eqref{eq:xing1} and~\eqref{eq:xing2}, one can upper-bound~\eqref{eq:mul1} as
\begin{align}
\PP_{\xi}\left(L(\xi_{\text{row}}^{(1)}) > L(\xi) \big| \Delta_1 \right)^{r_1} \le 2\exp(-n^{\frac{\epsilon}{2}}/4). \label{eq:rou1}
\end{align}
Similarly, one can upper-bound~\eqref{eq:mul2}  as
\begin{align}
\PP_{\xi}\left(L(\xi_{\text{row}}^{(\frac{n}{2}+1)}) > L(\xi) \big| \Delta_1\right)^{r_1} \le 2\exp(-n^{\frac{\epsilon}{2}}/4). \label{eq:rou2}
\end{align}
Finally, by~\eqref{eq:41},~\eqref{eq:mul1},~\eqref{eq:mul2},~\eqref{eq:rou1},~\eqref{eq:rou2}, and Lemma~\ref{lemma:strong}, we have that for sufficiently large $n$, the left-hand side (LHS) of~\eqref{eq:first} is upper-bounded by $5\exp(-n^{\frac{\epsilon}{2}}/4)$.

In a completely similar and symmetric fashion, one can also prove  inequality~\eqref{eq:second} in Lemma~\ref{lemma:10} (which focuses on the type class $\Xi_{\xi}(0,1)$). We omit it here for brevity. 
\end{proof}

\section{Proof Sketches of Theorems~\ref{thm:1} and~\ref{thm:2}} \label{sec:model2}
Due to the existence of atypical movies and different personalization parameters $\ta$ and $\tr$, proving the lower and upper bounds for Model 2 is more challenging. Fortunately, the key techniques used in Section~\ref{sec:simple} for both lower and upper bounds are still applicable to this more complicated model. Hence, for brevity, we respectively \emph{sketch} the proofs of Theorems~\ref{thm:1} and~\ref{thm:2} in Subsection~\ref{sec:thm1} and~\ref{sec:thm2} by highlighting the steps that are different from Section~\ref{sec:simple}. 

\subsection{Proof Sketch of Achievability (Theorem~\ref{thm:1})} \label{sec:thm1}
Suppose the sample probability $p$ satisfies inequality~\eqref{eq:achievability1} when $\ta \ne \tr$ or inequality~\eqref{eq:achievability2} when $\ta = \tr$. 
The achievability proof for Theorem~\ref{thm:1} also relies on the ML estimator $\phi_{\text{ML}}(\U, G_1,G_2)$  defined in~\eqref{eq:ml}. Recall from~\eqref{eq:error} that 
\begin{align}
P_{\text{err}}(\phi_{\text{ML}}) = \max_{\xi \in \Xi}\PP_\xi(\phi_{\text{ML}}(\U, G_1,G_2) \ne \xi). \label{eq:err}
\end{align}
We first consider a specific ground truth
$\xi \in \Xi$, and upper-bound the error probability $\PP_\xi(\phi_{\text{ML}}(\U, G_1,G_2) \ne \xi)$. The fraction of atypical movies in $\xi$ can be arbitrary.

Analogous to~\eqref{eq:sy2} in Section~\ref{sec:achievability}, we have
\begin{align}
\PP_\xi(\phi_{\text{ML}}(\U, G_1,G_2) \!\ne\! \xi) \!\le\! \sum_{\xi' \in \Xi: \xi' \ne \xi} \PP_{\xi}(L(\xi') \!\le\! L(\xi)). \label{eq:40}
\end{align}
Consider another instance $\xi' \in \Xi$ such that $\xi' \ne \xi$, and recall the definitions of $\Iaa,\Iar, \Ira, \Irr$ in~\eqref{eq:i1} and~\eqref{eq:i2}. Let 
\begin{align*}
&\saa \triangleq [n] \times \Iaa, \quad   \sar \triangleq [n] \times \Iar, \\
&\sra \triangleq [n] \times \Ira,  \quad \srr \triangleq [n] \times \Irr,
\end{align*}
respectively be the set of $(i,j)$-pairs such that movie $j$ belongs to $\Iaa,\Iar,\Ira, \Irr$. For $u, v \in \{\mathrm{a},\mathrm{r}\}$, we further partition $\mathcal{S}_{uv}$ into two subsets:
\begin{align*}
\mathcal{S}_{uv}^{\text{e}} &\triangleq \{(i,j)\in \mathcal{S}_{uv}: \bbij = \bij\}, \\
\mathcal{S}_{uv}^{\text{ue}} &\triangleq \{(i,j)\in \mathcal{S}_{uv}: \bbij \ne \bij\},
\end{align*}
where $\mathcal{S}_{uv}^{\text{e}}$ (resp. $\mathcal{S}_{uv}^{\text{ue}}$) contains all the pairs $(i,j) \in \mathcal{S}_{uv}$ such that the $(i,j)$-th entry in $B_{\xi}$ and the $(i,j)$-th entry in $B_{\xi'}$ are coincident (resp. different). One can adapt the calculations in~\eqref{eq:fan1}-\eqref{eq:fine} to Model 2 to obtain 
\begin{align}
&L(\xi) - L(\xi') = d_1 \Gamma_1 + d_2 \Gamma_2-  \!\!\!\sum_{u,v \in \{\mathrm{a},\mathrm{r}\}} \!\! \left(\Gamma^{\text{e}}_{uv} + \Gamma^{\text{ue}}_{uv} \right),  \notag \\
&\text{where} \ \Gamma^{\text{e}}_{uv}  \! \triangleq \!\!\!\sum_{(i,j) \in \mathcal{S}_{uv}^{\text{e}}} \!\!\!\! [f(\theta_u) \!-\! f(\theta_v)]\zij (1\!-\!\tij^{u}) + \zij\log\frac{\theta_u}{\theta_v}, \notag \\
&\Gamma^{\text{ue}}_{uv} \! \triangleq \sum_{(i,j) \in \mathcal{S}_{uv}^{\text{ue}}} \!\!\!\!\!\! - [f(\theta_u) \!+\! f(\theta_v)]\zij \tij^{u} +\zij \log\frac{1\!-\!\theta_u}{\theta_v}, \label{eq:42}
\end{align}
and $\Gamma_1$ and $\Gamma_2$ are respectively defined in~\eqref{eq:fan1} and~\eqref{eq:fine}.
Note that the expression above is parallel to~\eqref{eq:18} for Model 1. Applying  Chernoff bound (Lemma~\ref{lemma:generic})  with $t = 1/2$, we can upper-bound $\PP_{\xi}(L(\xi') \le L(\xi))$ by
\begin{align}
\exp\!\bigg\{\!\!\!-\!g_1I_{1}\frac{\log n}{n} \!-\! g_2I_{2}\frac{\log m}{m} \!-
\! \!\!\!\!\sum_{u,v \in \{\mathrm{a},\mathrm{r}\}}\!\!\!\! \left|\mathcal{S}_{uv}^{\text{e}}\right|  p\tau_{uv} \!+\! \left|\mathcal{S}_{uv}^{\text{ue}}\right| p\nu_{uv} \!\bigg\}. \label{eq:zhao}
\end{align}
To calculate the $|\mathcal{S}_{uv}^{\text{e}}|$ and $|\mathcal{S}_{uv}^{\text{ue}}|$, we define 
\begin{align*}
	t_{uv} \triangleq \sum_{j \in \mathcal{I}_{uv}} \mathbbm{1}\{\bbij \ne \bij \}, \quad \text{for} \ u, v \in \{\mathrm{a},\mathrm{r}\},
\end{align*}
where $i \in [n]$ is arbitrary. Then, one can show that
\begin{align*}
&\left|\suv\right| = \begin{cases} n(\frac{m}{2}-k_2), & \text{if } (u,v)=(\mathrm{a},\mathrm{a}), (\mathrm{r},\mathrm{r}), \\
nk_2,  & \text{if } (u,v)=(\mathrm{a},\mathrm{r}), (\mathrm{r},\mathrm{a}),
\end{cases}\\
&\left|\suv^{\text{ue}}\right| = \left(n-2k_1\right)t_{uv} + 2k_1\left(\left|\suv\right| - t_{uv}\right),
\end{align*}
and $\left|\suv^{\text{e}}\right| = \left|\suv\right|- \left|\suv^{\text{ue}}\right|$.
By routine calculations and using the fact that $\tarr \ge \tar$ for any $\ta,\tr \in (0,\frac{1}{2})$, we further upper-bound~\eqref{eq:zhao} by
\begin{align}
	&\exp\!\bigg\{\!-g_1I_{1}\frac{\log n}{n} - g_2I_{2}\frac{\log m}{m}- \Phi(k_1,k_2,t_{\mathrm{aa}}, t_{\mathrm{rr}}) \bigg\}, \label{eq:bound} \\
	&\text{where} \ \Phi(k_1,k_2,t_{\mathrm{aa}}, t_{\mathrm{rr}}) \!=\! 2\tar nk_2 \!+\! k_1\left(m\!-\!2k_2\right)(\taaa\!+\!\trrr) \notag \\
	&\qquad\qquad\qquad\qquad+ (t_{\mathrm{aa}} + t_{\mathrm{rr}})(n-4k_1)\min\{\taaa,\trrr\}. \notag
\end{align}
Note that~\eqref{eq:bound} depends only on   $k_1,k_2,t_{\mathrm{aa}}, t_{\mathrm{rr}}$. Analogous to the type class decomposition technique used in~\eqref{eq:ni}-\eqref{eq:t4}, we define $\Xi_{\xi}(k_1,k_2,t_{\mathrm{aa}},t_{\mathrm{rr}})$ as the set of $\xi'$ with the same error probability $\PP_{\xi}(L(\xi') \le L(\xi))$, and
\begin{align*}
&\mathcal{T}' \!\triangleq \! \left\{(t_{\mathrm{aa}},t_{\mathrm{rr}})\!: t_{\mathrm{aa}} \in \left[0:(m/2)\!-\!k_2\right], t_{\mathrm{rr}} \in \left[0:(m/2)\!-\!k_2\right] \right\}.
\end{align*}
For any $\epsilon > 0$, let the auxiliary parameter $\delta_{\epsilon} \triangleq \min\{\frac{\epsilon}{2I_2}, \frac{\epsilon}{8(1+\epsilon)} \}$. Recalling the definitions of $\mathcal{T}_1, \mathcal{T}_2, \mathcal{T}_3, \mathcal{T}_4$ in~\eqref{eq:t1}-\eqref{eq:t4}, we can bound the error probability in~\eqref{eq:40} by
\begin{align}
\sum_{i=1}^4 \sum_{(k_1,k_2) \in \mathcal{T}_i} \sum_{(t_{\mathrm{aa}},t_{\mathrm{rr}}) \in \mathcal{T}'} |\Xi_{\xi}(k_1,k_2,t_{\mathrm{aa}},t_{\mathrm{rr}})|  \PP_{\xi}(L(\xi') \le L(\xi)). \notag
\end{align}
It then remains to show that the error probabilities induced by each of the four subsets $\mathcal{T}_1$, $\mathcal{T}_2$, $\mathcal{T}_3$, and $\mathcal{T}_4$ vanish. 

\subsubsection{Case 1: $k_1 \le \delta_{\epsilon} n$ and $k_2 \le \delta_{\epsilon} m$}
Let 
\begin{align*}
&\Upsilon_1 \triangleq \left(1-2\delta_{\epsilon}\right)I_{1}(\log n) + pm\left(1-2\delta_{\epsilon}\right)(\taaa+\trrr), \quad \text{and}\\
&\Upsilon_2 \triangleq \left(1-2\delta_{\epsilon}\right)I_{2}(\log m) + 2pn \tar.
\end{align*}
In this case, the bound~\eqref{eq:bound} for $\PP_{\xi}\left(L(\xi') \le L(\xi) \right)$ can be further bounded   by
\begin{align}
&\exp\left\{-k_1\Upsilon_1 -k_2\Upsilon_2 -pn(t_{\mathrm{aa}}+t_{\mathrm{rr}})(1-4\delta_{\epsilon})\min\{\taaa,\trrr\} \right\}. \notag
\end{align}
Since $p$ satisfies the conditions in Theorem~\ref{thm:1}, we have 
\begin{align*}
&k_1\Upsilon_1 \ge k_1\left(2+\epsilon\right) \log n, \quad \text{and} \\
&k_2\Upsilon_2 +pn(t_{\mathrm{aa}}+t_{\mathrm{rr}})(1-4\delta_{\epsilon})\min\{\taaa,\trrr\} \\
&\qquad\qquad\qquad\ge \left(1+(\epsilon/2)\right)(2k_2 + t_{\mathrm{aa}} +t_{\mathrm{rr}}) \log m.
\end{align*}
Moreover, by noting that the size of  $\Xi_{\xi}(k_1,k_2,t_{\mathrm{aa}},t_{\mathrm{rr}})$ is
\begin{align}
& \binom{\frac{n}{2}}{k_1}^2\binom{\frac{m}{2}}{k_2}^2\binom{\frac{m}{2}-k_2}{t_{\mathrm{aa}}}\binom{\frac{m}{2}-k_2}{t_{\mathrm{rr}}} \sum_{t_{\mathrm{ar}}=0}^{k_2} \sum_{t_{\mathrm{ra}}=0}^{k_2} \binom{k_2}{t_{\mathrm{ar}}}\binom{k_2}{t_{\mathrm{ra}}} \notag \\
& \le \exp\left(2k_1 \log n\right) \exp\left((2k_2+t_{\mathrm{aa}}+t_{\mathrm{rr}}) \log m + 2k_2\right), \notag 
\end{align}
we eventually obtain that 
\begin{align*}
&\sum_{ \substack{(k_1,k_2)\\ \in \mathcal{T}_1}} \sum_{\substack{(t_{\mathrm{aa}},t_{\mathrm{rr}}) \\\in \mathcal{T}'}} \!\!\!|\Xi_{\xi}(k_1,k_2,t_{\mathrm{aa}},t_{\mathrm{rr}})|  \PP_{\xi}(L(\xi') \le L(\xi)) \!\le\! 16m^{-\epsilon},\\
&\sum_{\substack{(k_1,k_2) \\ \in \mathcal{T}_2}} \sum_{\substack{(t_{\mathrm{aa}},t_{\mathrm{rr}})\\ \in \mathcal{T}'}} \!\!\!|\Xi_{\xi}(k_1,k_2,t_{\mathrm{aa}},t_{\mathrm{rr}})|  \PP_{\xi}(L(\xi') \le L(\xi)) \!\le\! 32n^{-\epsilon}.
\end{align*}

\subsubsection{Case 2: $k_1 > \delta_{\epsilon} n$ or $k_2 > \delta_{\epsilon} m$}
In this case, one can apply similar techniques used in~\eqref{eq:case2} to show that the error probability $\PP_{\xi}\left(L(\xi') \le L(\xi) \right) \le e^{-\Omega(m\log m + n \log n)}$ and the size of $\Xi$ is bounded by $2^{n + 3m}$. Thus,
\begin{align}
&\sum_{(k_1,k_2)\in \mathcal{T}_3 \cup \mathcal{T}_4} \sum_{ (t_{\mathrm{aa}},t_{\mathrm{rr}})\in \mathcal{T}'} \!\! |\Xi_{\xi}(k_1,k_2,t_{\mathrm{aa}},t_{\mathrm{rr}})| \cdot \PP_{\xi}(L(\xi') \!\le\! L(\xi)) \notag \\
&\qquad\qquad\qquad\qquad\qquad\qquad\qquad\qquad\qquad= o(1). \label{eq:final}
\end{align}
Combining the analysis from~\eqref{eq:40} to~\eqref{eq:final}, we obtain that $\PP_\xi(\phi_{\text{ML}}(\U, G_1,G_2) \!\ne\! \xi) = o(1)$. By noting that the analysis from~\eqref{eq:40} to~\eqref{eq:final} is valid for any $\xi \in \Xi$ with an arbitrary fraction of atypical movies, one can eventually show that $$P_{\text{err}}(\phi_{\text{ML}}) = \max_{\xi \in \Xi}\PP_\xi(\phi_{\text{ML}}(\U, G_1,G_2) \ne \xi) = o(1).$$

\subsection{Proof Sketch of Converse (Theorem~\ref{thm:2})} \label{sec:thm2}

The following shows that for any $\epsilon > 0$, the error probability $P_{\text{err}} \to 1$ as $n \to \infty$ if the sample probability $p$ satisfies  the conditions in Theorem~\ref{thm:2}.

Similar to~\eqref{eq:sj1} in Subsection~\ref{sec:converse}, we first consider the ML estimator $\phi_{\mathrm{ML}}$ with respect to a specific ground truth $\xi$. Note that the success probability $P_{\text{suc}}(\xi)$ is upper-bounded by $\PP_{\xi}(\bigcap_{\xi' \in \Xi_{\xi}(k_1,k_2,t_{\mathrm{aa}},t_{\mathrm{rr}})}\{L(\xi') > L(\xi)\})$ for any type class $\Xi_{\xi}(k_1,k_2,t_{\mathrm{aa}},t_{\mathrm{rr}})$. The following lemma implies that the success probability $P_{\text{suc}}(\xi) \to 0$  even if we restrict our analysis to a specific type class $\Xi_{\xi}(k_1,k_2,t_{\mathrm{aa}},t_{\mathrm{rr}})$.

\begin{lemma}\label{lemma:1000}
Consider any ground truth $\xi \in \Xi$ and sufficiently large $n$.	(i) When $p < \frac{(2(1-\epsilon)-I_1)\log n}{(\taaa + \trrr)m}$, we have 
\begin{align} \PP_{\xi}\Bigg(\bigcap_{\xi' \in \Xi_{\xi}(1,0,0,0)} \!\!\{L(\xi') \!>\! L(\xi)\}\Bigg) \le 5\exp\Big(-\frac{1}{4}n^{\frac{\epsilon}{2}}\Big).\label{eq:1000}
\end{align}
(ii) When $p < \frac{(1-\epsilon)\log m}{\taaa n}$, we have
\begin{align} \PP_{\xi}\Bigg(\bigcap_{\xi' \in \Xi_{\xi}(0,0,1,0)} \!\!\{L(\xi') \!>\! L(\xi)\}\Bigg) \le \exp\Big(-\frac{1}{8}m^{\frac{\epsilon}{2}}\Big). \label{eq:0010}
\end{align}
(iii) When $p < \frac{(1-\epsilon)\log m}{\trrr n}$, we have
\begin{align} \PP_{\xi}\Bigg(\bigcap_{\xi' \in \Xi_{\xi}(0,0,0,1)} \!\!\{L(\xi') \!>\! L(\xi)\}\Bigg) \le \exp\Big(-\frac{1}{8}m^{\frac{\epsilon}{2}}\Big). \label{eq:0001}
\end{align}
(iv) When (a) $\ta \ne \tr$ and  $p < \frac{((1-\epsilon)-I_2)\log m}{2\tar n}$ or (b) $\ta = \tr$ and $I_2 \le 2(1-\epsilon)$, we have
\begin{align} \PP_{\xi}\Bigg(\bigcap_{\xi' \in \Xi_{\xi}(0,1,0,0)} \!\!\!\{L(\xi') \!>\! L(\xi)\}\Bigg) \le 5\exp\Big(-\frac{1}{4}m^{\frac{\epsilon}{2}}\Big). \label{eq:0100}
\end{align}
\end{lemma}

Note that Theorem~\ref{thm:2} follows directly from Lemma~\ref{lemma:1000}. Thus, it remains to prove Eqns.~\eqref{eq:1000}-\eqref{eq:0100} in Lemma~\ref{lemma:1000}.

\begin{proof}[Proof of Eqn.~\eqref{eq:1000}]
 Consider the type class $\Xi_{\xi}(1,0,0,0)$ with $k_1 = 1$ and $k_2 = t_{\mathrm{aa}} = t_{\mathrm{rr}} = 0$. Recalling the definitions of $\xi_{\text{row}}^{(i)}$, $r_1$, $\Delta_1$, and following the steps in~\eqref{eq:41}-\eqref{eq:mul2}, we can upper-bound the LHS of~\eqref{eq:1000} by  
\begin{align}
&\PP_{\xi}(\Delta_1^c) +  \PP_{\xi}\left(L(\xi_{\text{row}}^{(1)}) > L(\xi) \big| \Delta_1 \right)^{r_1} \notag \\
&\qquad\qquad\qquad\qquad+ \PP_{\xi}\left(L(\xi_{\text{row}}^{(\frac{n}{2}+1)}) > L(\xi) \big| \Delta_1 \right)^{r_1}. \label{eq:suc}
\end{align}
One can formulate 
$L(\xi) - L(\xi_{\text{row}}^{(1)})$ as per~\eqref{eq:42}, wherein $\Gamma_1, \Gamma_2, \Gamma_{uv}^{\text{ue}}$ are represented in terms of
\begin{align*}
&\s_{\mathrm{aa}}^{\mathrm{ue}} = \{(1,j)\}_{j=1}^{m/2}, \quad \s_{\mathrm{rr}}^{\mathrm{ue}} = \{(1,j)\}_{j=(m/2)+1}^{m},  \quad \text{and}\\
&\Imm = \left[2:\frac{n}{2}\right], \ \Imw = \emptyset, \ \Iwm = \{1\}, \ \Iww = \left[\frac{n}{2}+1:n\right].
\end{align*}
We apply Lemma~\ref{lemma:chernoff1} (with $K= (n/2)-1$, $L_1 = L_2 = m/2$, $\theta_1 = \ta$, $\theta_2 = \tr$) to lower-bound $\PP_{\xi}(L(\xi_{\text{row}}^{(1)}) \le L(\xi))$  by 
\begin{align*}
	\frac{1}{4}\! \exp\left\{\!\!-(1\!+\!o(1))\!\left(\frac{n}{2}\!-\!1\right)I_{1}\frac{\log n}{n} \!- \!\!\sum_{k=1}^2(1+o(1))\frac{m}{2} p h(\theta_k)\! \right\}\!,
\end{align*}
where $h(\theta_1) = h(\ta) = \taaa$ and $h(\theta_2) = h(\tr) = \trrr$. Following the derivations in~\eqref{eq:xing2} and~\eqref{eq:rou1} and noting that $p < \frac{(2(1-\epsilon)-I_1)\log n}{(\taaa + \trrr)m}$, we have
\begin{align*}
\PP_{\xi}\left(L(\xi_{\text{row}}^{(1)}) > L(\xi) \big| \Delta_1 \right)^{r_1} \le 2\exp\left(-n^{\frac{\epsilon}{2}}/4\right).
\end{align*}
One can similarly bound $\PP_{\xi}(L(\xi_{\text{row}}^{(\frac{n}{2}+1)}) > L(\xi) \big| \Delta_1 )^{r_1}$ by $2\exp\left(-n^{\frac{\epsilon}{2}}/4\right)$.  Thus, \eqref{eq:suc} can be upper-bounded by $5\exp\left(-n^{\frac{\epsilon}{2}}/4\right)$.
This completes the proof of Eqn.~\eqref{eq:1000}.
\end{proof}

\begin{proof}[Proof of Eqns.~\eqref{eq:0010} and~\eqref{eq:0001}]
Consider the type class $\Xi_{\xi}(0,0,1,0)$ with $t_{\mathrm{aa}} = 1$ and $k_1 = k_2  = t_{\mathrm{rr}} = 0$.
For each $\xi' \in \Xi_{\xi}(0,0,1,0)$, one can formulate 
$L(\xi) - L(\xi')$ as per~\eqref{eq:42}, wherein $|\s_{\mathrm{aa}}^{\mathrm{ue}}| = n$ and $|\s_{\mathrm{rr}}^{\mathrm{ue}}| = 0$. Applying Lemma~\ref{lemma:chernoff1} (with $L_1 = n, K = L_2 = 0, \theta_1 = \ta$), we have 
\begin{align}
	\PP_{\xi}(L(\xi') \le L(\xi)) \ge \frac{1}{4} \exp\left\{-(1+o(1))np h(\ta) \right\}. \label{eq:47}
\end{align}
Conditioned on the ground truth $\xi$, the events $\{L(\xi') > L(\xi)\}$ for different $\xi' \in \Xi_{\xi}(0,0,1,0)$ are mutually independent. Thus, the LHS of~\eqref{eq:0010} equals
\begin{align}
\prod_{\xi' \in \Xi_{\xi}(0,0,1,0)} \PP_{\xi}\left(L(\xi') > L(\xi) \right) \le \exp\left(-m^{\frac{\epsilon}{2}}/8\right), \notag
\end{align}
where the last inequality follows from~\eqref{eq:47} and the facts that $p < \frac{(1-\epsilon)\log m}{\taaa n}$ and $|\Xi_{\xi}(0,0,1,0)| = m/2$. 

The proof of Eqn.~\eqref{eq:0001} in Lemma~\ref{lemma:1000} is similar to that of Eqn.~\eqref{eq:0010}, hence we omit it here for brevity. 
\end{proof}

\begin{proof}[Proof of Eqn.~\eqref{eq:0100}]

Consider the type class $\Xi_{\xi}(0,1,0,0)$ with $k_2 = 1$ and $k_1 = t_{\mathrm{aa}} = t_{\mathrm{rr}} = 0$.
For ease of presentation, we suppose the ground truth $\xi$ satisfies $\xi_{\mathcal{A}} = [1:m/2]$ and $\xi_{\mathcal{R}} = [(m/2)+1:m]$. 

Let $r_2 \triangleq m/\log^2 m$, $\mathcal{V}'_2 \triangleq [1:2r_2]\cup [(m/2)+1: (m/2)+2r_2] \subseteq \mathcal{V}_2$ be a subset of movie nodes, and $\Delta_2$ be the event that the number of isolated nodes in $\mathcal{V}'_2$ is at least $3r_2$.

\begin{lemma}[Parallel to Lemma~\ref{lemma:strong}] \label{lemma:strong2}
	The probability  that $\Delta_2$ occurs is at least $1 - \exp\left(-\frac{\eta^2(\alpha_2+\beta_2)m^2}{\log^4 m}\right)$, where $\eta \in (0,1)$ is arbitrary.
\end{lemma}

For each $j \in \xi_{\mathcal{A}}$ and $j' \in \xi_{\mathcal{R}}$, we define several variants of $\xi$ as follows:

\noindent{1)} $\xi_{\text{col}}^{(j)}$ is identical to $\xi$ except that $(\xi_{\text{col}}^{(j)})_{\mathcal{A}} = \xi_{\mathcal{A}} \setminus \{j\}$ and $(\xi_{\text{col}}^{(i)})_{\mathcal{R}} = \xi_{\mathcal{R}} \cup \{j\}$, i.e., movie $j$ in $\xi_{\text{col}}^{(j)}$ is a romance movie.

\noindent{2)} $\xi_{\text{col}}^{(j')}$ to be identical to $\xi$ except that $(\xi_{\text{col}}^{(j')})_{\mathcal{A}} \!=\! \xi_{\mathcal{A}} \cup \{j'\}$ and $(\xi_{\text{col}}^{(j')})_{\mathcal{R}} \!=\! \xi_{\mathcal{R}} \!\!\setminus\!\! \{j'\}$, i.e., movie $j'$ in $\xi_{\text{col}}^{(j')}$ is an action  movie.

\noindent{3)} $\xi_{\text{col}}^{(j,j')}$ is identical to $\xi$ except that movie $j$ and  movie $j'$ in $\xi_{\text{col}}^{(j,j')}$ are respectively romance and action movies.

By definition, the type class $\Xi_{\xi}(0,1,0,0)$ is equivalent to the set  $\{\xi_{\text{col}}^{(j,j')}: j \in \xi_{\mathcal{A}}, j' \in \xi_{\mathcal{R}} \}$.
Conditioned on $\Delta_2$, one can find a subset $\xi_{P_2} \subset \xi_{\mathcal{A}}$ and a subset $\xi_{Q_2} \subset \xi_{\mathcal{R}}$ such that $|\xi_{P_2}| = |\xi_{Q_2}| = r_2$ and all the nodes in $\xi_{P_2} \cup \xi_{Q_2}$ are not connected to one another.

\noindent{\bf (i) When $\ta \ne \tr$:} Let $[\xi_{P_2}]_k$ and $[\xi_{Q_2}]_k$ respectively be the $k$-th elements of $\xi_{P_2}$ and $\xi_{Q_2}$, where $k \in [1:r_2]$. We define 
$$\xi_{P_2,Q_2} = \left\{\left([\xi_{P_2}]_k, [\xi_{Q_2}]_j \right)\right\}_{k=1}^{r_2}.$$ Note that the LHS of~\eqref{eq:0100} can be upper-bounded  by 
\begin{align}
&\PP_{\xi}\Bigg(\bigcap_{(j,j') \in \xi_{P_2,Q_2}} \left\{ L\left(\xi_{\text{col}}^{(j,j')}\right) > L(\xi) \right\} \Bigg) \notag \\
&\le \PP_{\xi}\Bigg(\bigcap_{(j,j') \in \xi_{P_2,Q_2}} \left\{ L\left(\xi_{\text{col}}^{(j,j')}\right) > L(\xi) \right\} \bigg| \Delta_2 \Bigg) + \PP_{\xi}(\Delta_2^c) \label{eq:sas}.
\end{align}
Without loss of generality, we assume $(1,(m/2)+1) \in \xi_{P_2,Q_2}$. The key observation is that conditioned on $\Delta_2$, the events $\{L(\xi_{\text{col}}^{(j',j)}) > L(\xi)\}$ for different $(j,j') \in \xi_{P_2,Q_2}$ are mutually independent, thus the first term in~\eqref{eq:sas} equals
\begin{align}
\PP_{\xi}\Big( L(\xi_{\text{col}}^{(1,\frac{m}{2}+1)}) > L(\xi) \Big| \Delta_2 \Big)^{r_2}.  \label{eq:50} 
\end{align}
Since $k_2 = 1$, the parameters $t_{\mathrm{ar}}$ and $t_{\mathrm{ra}}$ can be chosen to equal either $0$ or $1$, but in the following we consider the scenario in which both $t_{\mathrm{ar}}$ and $t_{\mathrm{ra}}$ equal zero. Recall that 
$L(\xi) - L(\xi_{\text{col}}^{(1,\frac{m}{2}+1)})$ can be formulated as per~\eqref{eq:42}, wherein $\Gamma_1, \Gamma_2, \Gamma_{uv}^{\text{e}}, \Gamma_{uv}^{\text{ue}}$ are represented in terms of
\begin{align*}
&\s_{\mathrm{ar}}^{\mathrm{e}} \!=\! \{(i,1)\}_{i=1}^{n}, \ \s_{\mathrm{ra}}^{\mathrm{e}} \!=\! \{(i,(m/2)+1)\}_{i=1}^{n},  \ \s_{\mathrm{ar}}^{\mathrm{ue}} \!=\! \s_{\mathrm{ra}}^{\mathrm{ue}} \!=\! \emptyset, \\
&\Iaa \!=\! \left[2:\frac{n}{2}\right], \ \Iar \!=\! \{1\}, \ \Ira \!=\! \left\{\frac{m}{2}\!+\!1\right\}, \ \Irr \!=\! \left[\frac{m}{2}\!+\!2:m\right].
\end{align*}
Applying similar techniques used for Lemma~\ref{lemma:chernoff1}, we have
\begin{align}
&\PP_{\xi}\left( L(\xi_{\text{col}}^{(1,\frac{m}{2}+1)}) \le L(\xi) \right) \notag \\
&\ge \frac{1}{4} \exp\left\{-(1\!+\!o(1))I_{2}(\log m) \!-\! (1\!+\!o(1)) 2np\tar  \right\}. \label{eq:51}
\end{align}
Combining~\eqref{eq:50},~\eqref{eq:51}, Lemma~\ref{lemma:strong2}, and the fact that $p < \frac{((1-\epsilon)-I_2)\log m}{2\tar n}$, we have that~\eqref{eq:sas} is upper-bounded by
\begin{align}
 \frac{\PP_{\xi}\left(L(\xi_{\text{col}}^{(1,\frac{m}{2}+1)}) > L(\xi)  \right)^{r_2}}{\PP_{\xi}(\Delta_2)^{r_2}} + \PP_{\xi}(\Delta_2^c) \le  3\exp\left(-\frac{m^{\frac{\epsilon}{2}}}{4}\right). \notag
\end{align}

\begin{remark}  \label{remark:challenge} {\em
	Note that the above analysis for $\ta \ne \tr$ is sub-optimal---the  number of events that are most likely  to cause errors is $|\xi_{P_2}| \times |\xi_{Q_2}| = \mathcal{O}(m^2/(\log m)^4)$; however, among them only $\mathcal{O}(m/(\log m)^2)$ independent events are extracted to $\xi_{P_2,Q_2}$, as shown in equation~\eqref{eq:sas}. Hence, a factor of two is lost in the converse part. Furthermore, due to the fact that $\ta \ne \tr$, the approach used in the proof of Lemma~\ref{lemma:1000} (i.e., split $\PP_{\xi}(\cap_{j \in \xi_{P_2}, j' \in \xi_{Q_2}} \{L(\xi_{\text{col}}^{(j,j')}) > L(\xi) \}\big| \Delta_2)$ into two individual terms as per~\eqref{eq:suc}) does not yield a tight converse either.}
\end{remark}

\noindent{\bf (ii) When $\ta = \tr$:}   Without loss of generality, we assume $1 \in \xi_{P_2}$ and $(m/2)+1 \in \xi_{Q_2}$. Similar to~\eqref{eq:suc}, one can bound the LHS of~\eqref{eq:0100} by
\begin{align}
&\PP_{\xi}(\Delta_2^c) +  \PP_{\xi}\left(L(\xi_{\text{col}}^{(1)}) > L(\xi) \big| \Delta_2 \right)^{r_2} \notag \\
&\qquad\qquad\qquad\quad + \PP_{\xi}\left(L(\xi_{\text{col}}^{(\frac{m}{2}+1)}) > L(\xi) \big| \Delta_2 \right)^{r_2}. \label{eq:suc2}
\end{align}
This is because conditioned on $\Delta_2$, the events $\{L(\xi_{\text{col}}^{(j)}) > L(\xi)\}$ for different $j \in \xi_{P_2}$ are mutually independent, and the events $\{L(\xi_{\text{col}}^{(j')}) > L(\xi)\}$ for different $j' \in \xi_{Q_2}$ are also mutually independent.
Also, by noting that  
$L(\xi) - L(\xi_{\text{col}}^{(1)})$ can be formulated as per~\eqref{eq:42}, wherein $\Gamma_1, \Gamma_2, \Gamma_{uv}^{\text{e}}, \Gamma_{uv}^{\text{ue}}$ are represented in terms of
\begin{align*}
&\s_{\mathrm{ar}}^{\mathrm{e}} \!=\! \{(i,1)\}_{i=1}^{n}, \ \s_{\mathrm{ra}}^{\mathrm{e}} \!=\! \s_{\mathrm{ar}}^{\mathrm{ue}} \!=\! \s_{\mathrm{ra}}^{\mathrm{ue}} \!=\! \emptyset, \quad \text{and} \\
&\Iaa \!=\! \left[2:\frac{n}{2}\right], \ \Iar \!=\! \{1\}, \ \Ira \!=\! \emptyset, \ \Irr \!=\! \left[\frac{m}{2}\!+\!2:m\right],
\end{align*}
we obtain 
\begin{align}
\PP_{\xi}\!\left(\! L(\xi_{\text{col}}^{(1)}) \!\le\! L(\xi) \!\right) \!\ge\! \frac{1}{4} \exp\left\{-(1\!+\!o(1))\frac{I_{2}\log m}{2} \right\}. \label{eq:55}
\end{align}
Similarly, we have
\begin{align}
\PP_{\xi}\!\left(\! L(\xi_{\text{col}}^{(\!\frac{m}{2}\!+\!1)}) \!\le\! L(\xi) \!\right) \!\ge\! \frac{1}{4} \exp\!\left\{\!-(1\!+\!o(1))\frac{I_{2}\log m}{2} \right\}. \label{eq:56}
\end{align}
Combining~\eqref{eq:suc2}-\eqref{eq:56}, Lemma~\ref{lemma:strong2}, and the fact that $I_2 \le 2(1-\epsilon)$, one can eventually show that the LHS of~\eqref{eq:0100} is bounded by $5\exp\left(-m^{\frac{\epsilon}{2}}/4\right)$. This completes the proof of Eqn.~\eqref{eq:0100}.\end{proof}

\section{Conclusion and future directions} \label{sec:conclusion}
This paper investigates two variants of a novel community recovery problem based on a partially observed rating matrix and social and movie graphs. Our information-theoretic characterizations on the sample complexity quantify the gains due to graph side-information; in particular, there exists a certain regime in which simultaneously observing two pieces of graph side-information is critical to reduce the sample complexity.

While the information-theoretic characterization for Model~2 is optimal in a certain parameter regime and order-optimal in the remaining parameter regime, one would expect that overcoming the challenge discussed in Remark~\ref{remark:challenge} and establishing a sharp threshold by filling the small gap for the regime in which our bounds do not match would be a fruitful endeavour.


\appendices

\section{Proof of Lemma~\ref{lemma:strong}}
	Let $N \triangleq 2 \binom{2r_1}{2} = 4r_1^2-2r_1$, $N' \triangleq 4r_1^2$, $\{X_i\}_{i=1}^N \stackrel{\text{i.i.d.}}{\sim} \text{Bern}(\alpha_1)$, $\{Y_i\}_{i=1}^{N'} \stackrel{\text{i.i.d.}}{\sim} \text{Bern}(\beta_1)$,
	and 
	$X \triangleq \sum_{i=1}^N X_i + \sum_{i=1}^{N'} Y_i$ 
	be the number of edges in $\mathcal{G}_1$. Thus, the number of non-isolated nodes is at most $2X$. Note that $\E(X) = N\alpha_1 + N'\beta_1$, which lies in the interval $[3r_1^2(\alpha_1+\beta_1), 4r_1^2(\alpha_1+\beta_1)]$ for sufficiently large $n$. For any $\eta \in (0,1)$, by applying the multiplicative Chernoff bound, we have 
	\begin{align*}
	\PP\left(X \ge
	\left(1+\eta\right)4r_1^2(\alpha_1 + \beta_1) \right) &\le \PP\left(X \ge \left(1+\eta\right)\E(X) \right) \\
	&\le \exp\left(-\frac{\eta^2(\alpha_1+\beta_1)n^2}{\log^4 n}\right).
	\end{align*}
	Therefore, with probability at least $1 - \exp\left(-\frac{\eta^2(\alpha_1+\beta_1)n^2}{\log^4 n}\right)$, $X \le \left(1+\eta\right)4r_1^2(\alpha_1 + \beta_1) < r_1/2$, and the number of non-isolated nodes is at most $r_1$.

\ifCLASSOPTIONcaptionsoff
  \newpage
\fi

\bibliographystyle{IEEEtran}
\bibliography{reference}

\end{document}